%% file: arxiv_final.tex
\documentclass[journal]{IEEEtran}

\input{myshorts}

\usepackage{amsfonts,amsthm,amsmath,amssymb}
\usepackage{pifont}
\newcommand{\cmark}{\ding{51}}%
\newcommand{\xmark}{\ding{55}}%
\usepackage{adjustbox}
\usepackage{graphicx}
\usepackage{graphicx,subcaption}
\usepackage[noadjust]{cite}
\usepackage[ruled,vlined]{algorithm2e}
\usepackage[dvipsnames]{xcolor} 
\usepackage{mathtools}
\DeclareMathAlphabet{\pazocal}{OMS}{zplm}{m}{n}
\DeclareMathOperator*{\argmin}{argmin}
\DeclareMathOperator*{\argmax}{argmax}
\DeclareMathOperator{\EX}{\mathbb{E}}
\newtheorem{Claim}{Claim}

\newtheorem{Theorem}{Theorem}
\newtheorem{corollary}{Corollary}
\newtheorem{definition}{Definition}
\usepackage{array, makecell}

\SetKwInput{KwInput}{Input}                
\SetKwInput{KwOutput}{Output}              

\begin{document}
\title{MAP Estimation of  Graph Signals}
	\author{
	  Guy Sagi  and Tirza Routtenberg, \IEEEmembership{Senior Member, IEEE} 
		\thanks{G. Sagi and T. Routtenberg are with the School of Electrical and Computer Engineering, Ben-Gurion University of the Negev, Beer Sheva, Israel.  Tirza Routtenberg is  also with the  Department of Electrical and Computer
Engineering, Princeton University, Princeton, NJ 08544 USA.
e-mail: sagix@post.bgu.ac.il, tirzar@bgu.ac.il. This work is partially supported
by THE ISRAEL SCIENCE FOUNDATION (grant No. 1148/22) and by
the  Israeli Ministry of National Infrastructure, Energy, and Water Resources.\\
© 20XX IEEE.  Personal use of this material is permitted.  Permission from IEEE must be obtained for all other uses, in any current or future media, including reprinting/republishing this material for advertising or promotional purposes, creating new collective works, for resale or redistribution to servers or lists, or reuse of any copyrighted component of this work in other works.}
		}

	\maketitle

\begin{abstract}
In this paper,
we consider the problem of recovering random graph signals from nonlinear measurements. 
We formulate the  maximum
{\em a-posteriori} probability (MAP) estimator, which results in a 
nonconvex
optimization problem. 
Conventional iterative methods for minimizing nonconvex problems are sensitive to the initialization, have high computational complexity, and do not utilize the underlying graph structure behind the data.
In this paper we propose three new estimators that are based on the Gauss-Newton method:
1) the elementwise graph-frequency-domain MAP (eGFD-MAP) estimator;
2) the sample graph signal processing MAP (sGSP-MAP) estimator; and 3) the GSP-MAP estimator. At each iteration, these estimators are updated by the outputs of two graph filters, with the previous state estimator and the residual as
the input graph signals.
The eGFD-MAP estimator is based on neglecting the mixed-derivatives of different graph frequencies in the Jacobian matrix and the off-diagonal elements in the covariance matrices. 
Consequently, it updates the elements of the graph signal in the graph-frequency domain independently, which reduces the computational complexity compared to the conventional MAP estimator.
The sGSP-MAP and GSP-MAP estimators are based on optimizing the graph filters at each iteration of the   Gauss-Newton algorithm.
We state conditions under which the new estimators coincide with the MAP estimator in the case of an observation model with orthogonal graph frequencies.
We evaluate the performance of the estimators for nonlinear graph signal recovery tasks, both with synthetic data and with the real-world problem of state estimation in power systems.
These simulations show the advantages of the proposed estimators in terms of computational complexity, mean-squared-error, and robustness to the initialization of the algorithms.
\end{abstract}

\begin{keywords}
Graph signal processing (GSP),
graph filters,
graph-frequency domain,
nonlinear estimation,
maximum
{\em a-posteriori} probability (MAP) estimation
\end{keywords}

\section{Introduction}
The emerging field of graph signal processing (GSP) deals with processing data indexed by general graphs with concepts and techniques inspired by traditional digital signal processing (DSP). These techniques include graph Fourier transform (GFT), graph filter designs \cite{8347162,Shuman_Ortega_2013,Isufi_Leus2017,6319640}, and the sampling and recovery of graph signals \cite{9244650,6854325,7352352}. Most of these techniques have been used for various tasks under a linear measurement model. However, modern networks are often large and complex, contain heterogeneous datasets, and are characterized by nonlinear models \cite{9343697,8347160}, and as a result, graph signals are often difficult to recover in these networks.
Examples of applications with such networks include brain network connectivity \cite{shen2016nonlinear}, environmental monitoring \cite{8496842}, and power flow equations in power systems \cite{drayer2018detection,Grotas2019,shaked2021identification}.
Hence, the development of GSP methods for the estimation of graph signals in nonlinear models has considerable practical significance.

The mean-squared-error (MSE) is one of the most commonly-used criteria
of accuracy for estimation and reconstruction purposes. In Bayesian estimation, 
the optimal minimum MSE (MMSE)  estimator 
usually lacks a closed-form expression in nonlinear models and is often computationally intractable. 
Therefore,   the linear MMSE (LMMSE) estimator and other
low-complexity estimators  (see e.g. \cite{Edfors_Borjesson_1996,dd,Berman_letter}) are widely used in practice. 
The linear 
 GSP-LMMSE estimator, which minimizes the MSE among estimators that are represented as an output of a graph filter, has been suggested in \cite{kroizer2021bayesian} for the recovery of graph signals, and its properties are discussed therein. In particular, it has been shown that
 the GSP-LMMSE estimator has low computational complexity and the ability to adapt to changes in the graph topology. In addition, it has been extended to the widely-linear estimation of complex-valued signals in \cite{Amar_Routtenberg}.
However,  linear estimators 
fail to take into account the nonlinearity of the system, which
may lead to degraded performance compared with nonlinear methods. 
Consequently, developing nonlinear estimation methods that take advantage of the graph structure and GSP theory have the potential to significantly improve the estimation performance in these cases.

Nonlinear methods can significantly outperform linear estimators in terms of MSE.
For example, the superiority of the nonlinear maximum
{\em a-posteriori} probability (MAP) estimator compared with the LMMSE estimator in nonlinear filtering is well known \cite{Fatemi_Svensson_Morelande2012}.
In order to implement
nonlinear methods, such as 
the MAP estimator, in nonlinear settings,  iterative approaches are used.
In particular, the Gauss-Newton method is commonly used to find the MAP \cite{Fatemi_Svensson_Morelande2012} and other nonlinear estimators in various applications \cite{Abur_book,Monticelli_2000,Cosovic_Vukobratovic_2019,Mensing_Plass_2006,Stoica_Moses_Friedlande1989,Bell_Cathey1993,schweiger2005gauss}.
While the Gauss-Newton method provides fast convergence  and accurate estimates for 
proper initial values 
\cite{bjorck1996numerical,Li_Scaglione2013}, it
is sensitive to initialization, and may converge to local minima or even diverge \cite{Blaschke1997}. 
Thus, integrating the graph information, e.g. by using graph filters,  has the potential to enhance the robustness of iterative implementations of the MAP estimator. 

Graph filters have been used for many signal processing tasks, such as denoising \cite{7032244,ZHANG20083328}, classification \cite{6778068}, and anomaly detection \cite{drayer2018detection}. 
Model-based recovery of  graph signals  by GSP filters for linear models was treated in \cite{7117446,7032244,Isufi_Leus2017,7891646}.
Nonlinear graph filters were considered in \cite{8496842}, but they require higher-order statistics that are not completely specified in the general nonlinear case. Recently,
graph neural network approaches were considered in \cite{Isufi_Ribeiro,gama2020graphs}.
In addition, 
  a Gauss-Newton unrolled neural network method was developed in \cite{Yang_Giannakis_Sun2020}  for the application of power system state estimation (PSSE). 
However, data-based methods necessitate extensive stationary training sets, and do not necessarily utilize the model information. Using the model at hand, one can design estimators  with improved performance in terms of MSE, interpretability,
robustness, complexity, and flexibility.
Fitting graph-based models to given data was considered in \cite{7763882,Hua_Sayed_2020,confPaper}.  However, model-fitting approaches minimize the modeling error, and in general have worse performance than estimators that minimize the estimation error directly \cite{lecture_notes}. 

In this paper, we consider the nonlinear estimation of random graph signals with a nonlinear Gaussian observation model.
First, we present the MAP estimator, as well as its implementations in the vertex and in the graph-frequency domains by using the Gauss-Newton method. 
Then, we propose three new GSP estimators: 1) the elementwise graph-frequency-domain MAP (eGFD-MAP) estimator; 2) the sample GSP-MAP (sGSP-MAP) estimator; and 3) the GSP-MAP estimator. The eGFD-MAP estimator updates the coordinates of the estimator in the graph-frequency domain separately,
and has significantly lower computational complexity than the MAP estimator. The sGSP-MAP and the GSP-MAP estimators are based on optimizing the graph filters at each iteration of the Gauss-Newton algorithm.
We show that for 
models with measurement functions that have orthogonal graph frequencies, i.e. separable in the graph-frequency domain, 
  the eGFD-MAP, sGSP-MAP, and the GSP-MAP estimators coincide with the MAP estimator.
  We perform numerical simulations for: 1) synthetic data with orthogonal graph frequencies; and 2)  PSSE.
For the first case, it is shown that the eGFD-MAP, sGSP-MAP, and the GSP-MAP estimators achieve the same MSE as the MAP estimator, while the eGFD-MAP estimator significantly reduces the computational complexity, especially in large networks where the
other
estimators are intractable.
For the PSSE simulations, it is shown that the sGSP-MAP attains the MSE of the MAP estimator. Moreover, the eGFD-MAP and the GSP-MAP estimators have performance close to that of the MAP estimator
and outperform the linear estimator,
where the eGFD-MAP and GSP-MAP estimators are more robust to perturbed initialization than the MAP estimator.

The rest of this paper is organized as follows. In Section \ref{background_sec} we introduce the basics of GSP required for this paper.
In Section \ref{MAP_objectives_sec}, we formulate the estimation problem,  present the MAP 
estimation approaches (in the vertex and graph-frequency domains), and  describe the  Gauss-Newton implementation of the MAP 
estimators.
In Section \ref{eGFD_MAP_and_GSP_MAP}, we develop the eGFD-MAP, sGSP-MAP,  and GSP-MAP estimators and discuss their properties.
Simulations are presented in Section \ref{simulation}.
Finally, 
 conclusions are outlined  in Section \ref{conclusion}.

In the following, we denote vectors by boldface lowercase letters and matrices by boldface uppercase letters.
The operators $(\cdot)^T$ and $(\cdot)^{-1}$ denote the transpose and inverse,  respectively. The vector $\onevec$ denotes a vector of all ones, and  $\circ$ denotes the Hadamard product. The $m$th element of the vector $\avec$ is denoted by $a_m$ or $[\avec]_m$. The $(m,q)$th element of the matrix $\Amat$ is written as $A_{m,q}$ or $[\Amat]_{m,q}$.
For a vector $\avec$, ${\text{diag}}(\avec)$ is a diagonal matrix whose $i$th diagonal entry is $a_{i}$; when applied to a matrix, ${\text{diag}}(\Amat)$ is a vector collecting the diagonal elements of $\Amat$. In addition, ${\text{ddiag}}(\Amat)={\text{diag}}({\text{diag}}(\Amat))$ is the diagonal matrix whose entries on the diagonal are those
of $\Amat$.
The identity matrix and the zero vector are written as $\Imat$ and $\zerovec$, respectively.
The cross-covariance matrix of the vectors $\avec$ and $\bvec$ is denoted by $\Cmat_{\avec\bvec}\triangleq \EX[(\avec-\EX[\avec]) (\bvec-\EX[\bvec])^T]$.
The Jacobian matrix of a vector function $\gvec(\xvec)$, $\nabla_{\xvec}\gvec(\xvec)$, is a matrix in $\mathbb{R}^{K\times M}$, with the $(k,m)$th element equal to $\frac{\partial g_k}{\partial x_m}$, where $\gvec=\left[g_1,\ldots,g_K\right]^T$ and $\xvec=\left[x_1,\ldots,x_M\right]^T$, and
 $\nabla_\xvec^T  \gvec(\xvec) \define (\nabla_\xvec  \gvec(\xvec))^T$.
For a scalar function, $g(\xvec)$, we denote
$\nabla_\xvec^2 g(\xvec) \define\nabla_\xvec \nabla_\xvec^T  g(\xvec)$. Finally,  
$||\cdot||$ is the Euclidean norm.

\section{Background: Graph Signal Processing (GSP)}
\label{background_sec}
GSP is an emerging field that deals with developing signal processing methods for representing, analyzing, and processing signals that lie on a graph \cite{8347162,Shuman_Ortega_2013}.
Consider an undirected, connected, weighted graph ${\pazocal{G}}({\pazocal{V}},\xi,\Wmat)$, where $\pazocal{V}$ and $\xi$ are sets of vertices and edges, respectively. 
The matrix $\Wmat \in \mathbb{R}^{N \times N}$ is the nonnegative weighted adjacency matrix of the graph, where $N \define |\pazocal{V}|$ is the number of vertices in the graph. If there is an edge connecting vertices $i$ and $j$, i.e. 
 if $(i, j) \in \xi$,  the entry $\Wmat_{i,j}$ represents the weight of the edge; otherwise, $\Wmat_{i,j} = 0$.
The Laplacian matrix of the given graph ${\pazocal{G}}({\pazocal{V}},\xi,\Wmat)$, which  contains the information
on the graph structure, is defined by
\begin{equation}
    \Lmat \define \text{diag}\left(\Wmat\onevec\right) - \Wmat.
    \end{equation}
    Thus, 
the Laplacian matrix, $\Lmat$, is a real and positive semidefinite matrix
that satisfies
the null-space property, $\Lmat\onevec_M=\zerovec$,
and has nonpositive off-diagonal elements. In particular, 
its associated eigenvalue decomposition  is given by
 \begin{equation}
\Lmat = \Vmat\Lambdamat \Vmat^{T},   \label{SVD_new_eq}
 \end{equation}
where $\Lambdamat$ is a diagonal matrix consisting of the eigenvalues of 
$\Lmat$, $0= \lambda_1 < \lambda_2 \leq \ldots \leq \lambda_N $, $\Vmat$ is a matrix whose $n$th column, $\vvec_N$, is the eigenvector of $\Lmat$ that is associated with $\lambda_n$, and $\Vmat^{T}=\Vmat^{-1}$.

In this paper, a {\em{graph signal}} is 
an $N$-dimensional vector, $\avec$, that assigns a scalar value to each vertex, i.e. each entry $a_n$ denotes the signal value at vertex $n$, for $n=1,\ldots, N$.
The GFT of the graph signal $\avec$ is defined as \cite{Shuman_Ortega_2013}
 \begin{equation}
\label{GFT}
\tilde{\avec} \triangleq \Vmat^{T}\avec. 
 \end{equation}
Similarly, the inverse GFT (IGFT) of $\tilde{\avec}$ is given by $\Vmat\tilde{\avec}$.
The GFT with respect to (w.r.t.)  the graph Laplacian describes the variation in a graph signal while taking into account the underlying connectivity.
A graph signal is a graph-bandlimited signal with cutoff graph frequency $N_s$ if it satisfies \cite{8347162}
\begin{equation} \label{bandlimited_def}
    \tilde{a}_n =0,~ n =N_s+1,\dots,N,
\end{equation}
 which implies sparsity
in the signal's representation in the spectral, graph-frequency domain. Intuitively, similar to signals bandlimited in the discrete Fourier transform (DFT) domain,  a graph-bandlimited signal, satisfying \eqref{bandlimited_def}, has a small variation across neighboring nodes with small weights \cite{dabush2023verifying}. 

 Graph filters are useful tools for various GSP tasks. 
Linear and shift-invariant graph filters w.r.t. the graph shift operator (GSO) play essential roles in GSP. 
A graph filter is a function $f(\cdot)$ applied to a GSO, where here we use the Laplacian $\Lmat$ as the GSO, which allows the following eigendecomposition \cite{8347162}:
 \begin{equation} \label{laplacian_graph_filter}
  	f(\Lmat)= \Vmat f(\Lambdamat)\Vmat^T,
  \end{equation}
 where $f(\Lambdamat)$ is a diagonal matrix.
That is, $f(\lambda_n)$ is the graph frequency response of the filter at graph frequency $\lambda_n$, $n=1,\dots,N$,
and $f(\Lmat)$ is  diagonalized by the eigenvector matrix of  $\Lmat$, $\Vmat$.
We assume that the graph filter,  $f(\cdot)$, is a well-defined function on the spectrum of $\Lmat$, $\{\lambda_1,\ldots,\lambda_N\}$.

\section{MAP estimator of graph signals}
\label{MAP_objectives_sec}
In this section we formulate the MAP estimator for the problem of estimating a random graph signal by observing its noisy nonlinear function. First, the measurement model and assumptions are introduced in Subsection \ref{model_sub_sec}.
Then, we derive the MAP estimator in both the vertex and the graph-frequency domains in Subsection \ref{MAP_subsec}. In Subsection \ref{iterative_subsection}, we show the implementation of the MAP estimator by the Gauss-Newton method. 
\subsection{Model}
\label{model_sub_sec}
Consider the problem of recovering a random input graph signal, $\xvec\in\mathbb{R}^N$, based on the following nonlinear measurement model: 
\begin{equation} \label{Model}
	\yvec = \gvec(\Lmat,\xvec) + \wvec,
\end{equation}
where the measurement function, $\gvec:\mathbb{R}^{N \times N} \times{\mathbb{R}}^N \rightarrow {\mathbb{R}}^N$, and the Laplacian matrix, $\Lmat$, which represents the influence of the graph topology, are assumed to be known. 
In addition, it is assumed that 
$ \gvec(\Lmat,\xvec)$ is
continuously differentiable w.r.t. $\xvec$.
We assume that
the graph input signal,  $\xvec$, is  a Gaussian vector with mean $\muvec_\xvec$ and covariance matrix ${\Cmat}_{\xvec\xvec}$. The noise, $\wvec$,  is assumed to be a zero-mean Gaussian vector with covariance matrix $\Cmat_{\wvec\wvec}$. Finally, we assume that $\wvec$ and $\xvec$ are independent.
The nonlinear measurement model in \eqref{Model} is adequate for many applications within networked data and the Internet-of-Things (IoT) 
\cite{sahu2018communication,Sahu_Kar_Moura_Poor2016,kroizer2021bayesian}. In particular, this
model arises in PSSE, which is discussed in Section \ref{simulation}.

\subsection{MAP estimator}
\label{MAP_subsec}
In the following, we develop the MAP estimator for the measurement model in \eqref{Model}. 
The MAP estimator of $\xvec$ from $\yvec$ is given by
 \beqna
 \label{MAP_estimator1}
 \hat{\xvec}= \argmax_{{\xvec} \in \mathbb{R}^N}
  f(\xvec|\yvec)
  =\argmax_{{\xvec} \in \mathbb{R}^N}
 \log f(\yvec|\xvec)+\log f(\xvec),
\eeqna
where the second equality  is obtained by 
using Bayes’s rule,
applying the monotonically-increasing  logarithm function and removing constant terms w.r.t. $\xvec$.
By substituting  in
\eqref{MAP_estimator1} the considered model,  in which $\yvec$ given $\xvec$ is also a Gaussian vector with mean $\gvec(\Lmat,\xvec)$ and covariance ${\Cmat}_{\wvec\wvec}$, 
we obtain
 \beqna
 \label{MAP_estimator2}
 \hat{\xvec}=\argmin_{{\xvec} \in \mathbb{R}^N} Q(\xvec),
 \eeqna
where, after removing constant terms w.r.t. $\xvec$, we have
\beqna
\label{Q_def}
Q(\xvec)\define
\frac{1}{2}(\xvec-\muvec_\xvec)^{T}{\Cmat}_{\xvec\xvec}^{-1}(\xvec-\muvec_\xvec)\hspace{1.75cm}\nonumber\\ + \frac{1}{2}(\yvec-\gvec(\Lmat,\xvec))^{T}{\Cmat}_{\wvec\wvec}^{-1}(\yvec-\gvec(\Lmat,\xvec)).
\eeqna
The left term of the objective function on the r.h.s. of \eqref{Q_def} corresponds to
 the prior, and the right term corresponds to the noisy measurement model.
This objective function can be minimized by iterative algorithms to approximate the MAP estimator, as described in Section \ref{iterative_subsection}.
The objective function in \eqref{Q_def} discards information about the relationship between the graph signal and its underlying graph structure. As a result, it is less robust to perturbations of the initialization that are due to changes in the graph topology.
In addition, the MAP estimator only uses the measurement function and does not exploit additional GSP information on the graph signal, such as smoothness or graph-bandlimitness, that could be utilized to improve estimation performance.

In general, representing data in the graph-frequency domain can yield substantial data reduction, and minimize the computational requirements and memory use.
Thus, as a first step, we suggest to transform  $Q(\xvec)$ into its graph-frequency domain representation (as a function of $\tilde{\xvec}$):
\beqna
\label{Q_def_GSP}
Q_{freq}(\tilde{\xvec})\define
\frac{1}{2}(\tilde{\xvec}-\tilde{\muvec}_\xvec)^{T}{\Cmat}_{\tilde{\xvec}\tilde{\xvec}}^{-1}(\tilde{\xvec}-\tilde{\muvec}_\xvec)\hspace{1.5cm}\nonumber\\ + \frac{1}{2}(\tilde{\yvec}-\tilde{\gvec}(\Lmat,\Vmat\tilde\xvec))^{T}{\Cmat}_{\tilde\wvec\tilde\wvec}^{-1}(\tilde\yvec-\tilde\gvec(\Lmat,\Vmat\tilde\xvec)),
\eeqna
where $\tilde\yvec$, $\tilde\xvec$, $\tilde{\muvec}_\xvec$, and $\tilde\gvec$ are the GFT representations of $\yvec$, $\xvec$, ${\muvec}_\xvec$, and $\gvec$, respectively, as defined in \eqref{GFT}, such that
 $\tilde{\gvec}(\Lmat,\Vmat\tilde{\xvec})=\Vmat^T\gvec(\Lmat,\xvec)$. In addition, we use the fact that 
 \begin{equation*}
     {\Cmat}_{\tilde\xvec\tilde\xvec}\define {\rm{E}}\left[\Vmat^{T}(\xvec-\muvec_{\xvec})(\xvec-\muvec_{\xvec})^{T}\Vmat\right]=\Vmat^{T}{\Cmat}_{\xvec\xvec}\Vmat,
 \end{equation*}
 and similarly ${\Cmat}_{\tilde\wvec\tilde\wvec}=\Vmat^{T}{\Cmat}_{\wvec\wvec}\Vmat$. It can be verified that the r.h.s. of \eqref{Q_def} and the r.h.s. of \eqref{Q_def_GSP} are identical. Thus, these two objective functions will lead to the same final estimator, 
where in \eqref{Q_def} the update is in the vertex domain ($\xvec$),  and  in \eqref{Q_def_GSP} the update is in the graph-frequency domain ($\tilde{\xvec}$).
However, the efficiency and the convergence rate of the specific implementation in each domain may be different.
For example, in our simulations (Subsection \ref{simulation}), the implementation of the MAP estimator in the graph-frequency domain is much faster.


\subsection{Implementation by Gauss-Newton method}
\label{iterative_subsection}
Since $\gvec(\Lmat,\xvec)$ is a nonlinear and nonconvex function, direct minimization of the  objective functions in 
\eqref{Q_def} or \eqref{Q_def_GSP}
is intractable. 
Numerous algorithms have been proposed to minimize nonconvex objectives. Here, we  implement the 
Gauss-Newton  method, which is widely
employed to solve 
nonlinear weighted least squares (WLS) problems and to find the MAP estimator.
The Gauss-Newton method has 
a quadratic rate of convergence under suitable assumptions \cite{Blaschke1997}.

{\bf{Initialization:}} For all the algorithms described in this paper, the estimators can be initialized by 
  the prior mean,   i.e.  $\hat{\xvec}^{(0)}=\muvec_\xvec$ and $\hat{\tilde{\xvec}}^{(0)}=\tilde{\muvec}_\xvec$, where $\tilde{\muvec}_\xvec\define \Vmat^{T}\muvec_\xvec$.
  Alternatively, if training data is available or can be generated, it is possible to initialize the estimators by the LMMSE or the GSP-LMMSE \cite{kroizer2021bayesian} estimators.  
  In addition, 
 the stopping condition is attained when two successive MAP estimates, $\hat{\xvec}^{(t+1)}$
and $\hat{\xvec}^{(t)}$, are sufficiently close, where $t$ is the iteration index.

\subsubsection{Conventional MAP estimator}
\label{subsection_MAP}
 In order to locally minimize the MAP objective function in  \eqref{Q_def},
we use  the
Gauss-Newton method, which relies on Taylor's expansion to
linearize the measurement function, and iteratively updates the estimators until convergence (see Sec. 1.5.1 in \cite{BERTSEKAS_1999}). Specifically, 
 under the assumption that  the given estimator $\hat{{\xvec}}^{(t)}$ is close enough to  ${\xvec}$,
the first-order approximation of $\gvec(\Lmat,\xvec)$ is 
 \be
 \label{g_approx}
 \gvec(\Lmat,\xvec)\approx
 \gvec(\Lmat,\hat{{\xvec}}^{(t)}) + \Gmat(\Lmat,\hat{{\xvec}}^{(t)})(\xvec-\hat{{\xvec}}^{(t)}),
 \ee
 where 
\be
\label{GGG}
\Gmat(\Lmat,\hat{{\xvec}}^{(t)})\define \left.\nabla_\xvec \gvec(\Lmat,\xvec)\right|_{\xvec=\hat{{\xvec}}^{(t)}}
\ee
is the $N\times N$ Jacobian  matrix of the measurement function $\gvec(\Lmat,\xvec)$ evaluated at
$\hat{{\xvec}}^{(t)}$.
By substituting  \eqref{g_approx} in \eqref{Q_def}, we obtain the approximated linearized objective function:
\beqna
\label{Q_def_approx}
Q(\xvec)\approx Q_{lin}(\xvec,\hat{{\xvec}}^{(t)})
\hspace{4.75cm}\nonumber\\\define
\frac{1}{2}(\xvec-\muvec_\xvec)^{T}{\Cmat}_{\xvec\xvec}^{-1}(\xvec-\muvec_\xvec)\hspace{3.5cm}\nonumber\\ + \frac{1}{2}(\yvec-\gvec(\Lmat,\hat{{\xvec}}^{(t)}) - \Gmat(\Lmat,\hat{{\xvec}}^{(t)})(\xvec-\hat{{\xvec}}^{(t)}))^{T}{\Cmat}_{\wvec\wvec}^{-1}
\hspace{0.275cm}\nonumber\\\times
(\yvec-\gvec(\Lmat,\hat{{\xvec}}^{(t)}) - \Gmat(\Lmat,\hat{{\xvec}}^{(t)})(\xvec-\hat{{\xvec}}^{(t)})).\hspace{1.2cm}
\eeqna
The
Gauss-Newton method finds the next iterate $t+1$ by minimizing $ Q_{lin}(\xvec,\hat{{\xvec}}^{(t)})$ w.r.t. $\xvec$, which results in
\beqna\label{iteration_Newton2}
\hat{\xvec}^{(t+1)}=\arg\min_{\xvec\in\mathbb{R}^N} Q_{lin}(\xvec,\hat{{\xvec}}^{(t)})\hspace{3.65cm}\nonumber\\=
\hat{\xvec}^{(t)}-\alpha^{(t)}\left ({\Cmat}_{\xvec\xvec}^{-1}  + \Gmat^{T}(\Lmat,\hat{\xvec}^{(t)}){\Cmat}_{\wvec\wvec}^{-1}\Gmat(\Lmat,\hat{\xvec}^{(t)})\right)^{-1}
\nonumber\\
\times\Big({\Cmat}_{\xvec\xvec}^{-1} (\hat{\xvec}^{(t)}- \muvec_\xvec)\hspace{4.5cm}\nonumber\\
- \Gmat^{T}(\Lmat,\hat{\xvec}^{(t)}){\Cmat}_{\wvec\wvec}^{-1}(\yvec-\gvec(\Lmat,\hat{\xvec}^{(t)}))\Big),\hspace{1.85cm}
\eeqna
with $\alpha^{(t)}=1$. 

In practice, wisely setting the step size,
$\alpha^{(t)}\in(0,1]$, can improve the convergence rate \cite{Bell_Cathey1993}. 
In this paper, we compute the step size by a backtracking 
line search (Algorithm \ref{line_search_stages}). In this strategy, we iteratively reduce the step size, $\alpha^{(t)}$, until $\hat{{\xvec}}^{(t+1)}$ from \eqref{iteration_Newton2} with the tested step size that satisfies 
\begin{equation}\label{line_search1}
Q(\hat{{\xvec}}^{(t)})-Q(\hat{{\xvec}}^{(t+1)})>\Delta |Q(\hat{{\xvec}}^{(t)})|,
\end{equation}
where $\Delta\in\mathbb{R}$ satisfies $\Delta<<1$.
\vspace{-0.25cm}
\begin{algorithm}[hbt]
\SetAlgoLined
	\KwInput{ 
	\begin{itemize}
	\item Current estimator $\hat{\xvec}^{(t)}$
		\item Initial step size $\alpha_{0}\in (0,1]$, tuning parameter $\gamma\in(0,1)$ 
		\item Objective function $Q$
	\end{itemize}}
	{\textbf{Algorithm Steps:}}
	\begin{enumerate}
		\item Compute: $Q(\hat{{\xvec}}^{(t)})$
    	\item 
		\For{$k=0,1,\ldots K_{\max}$}{
		\begin{itemize} \item 
		Compute
		${\hat{\xvec}}^{(t+1)}$ 
		with  $\alpha^{(t)}=\alpha_{k}$
		\item Compute $Q(\hat{{\xvec}}^{(t+1)})$
		\end{itemize}
\If{\eqref{line_search1} does not hold}
{\textcolor{black}{$\alpha_{k+1}=\left\{ \begin{array}{lr} \gamma\alpha_{k} & k< K_{\max}\\
0 & k= K_{\max}\end{array}\right.$}
}\Else{\text{break}}
}
	\end{enumerate}
\KwOutput{ Step size: $\alpha^{(t)} =\alpha_k$ }
	\caption{Backtracking line search}
	\label{line_search_stages}
\end{algorithm}

\subsubsection{MAP estimator in the graph-frequency domain} \label{GSP_estimator_section}
In this subsection,
we evaluate the graph-frequency-domain update of the MAP
estimator.
Similar to the derivation of \eqref{iteration_Newton2}, the minimization of
$Q_{freq}(\tilde{\xvec})$ from \eqref{Q_def_GSP} w.r.t. $\tilde{\xvec}$ by the  Gauss-Newton  method  results in the following update equation:
\beqna\label{iteration_Newton2_GSP}
\hat{\tilde{\xvec}}^{(t+1)}=
\hat{\tilde{\xvec}}^{(t)}\hspace{6.25cm}\nonumber\\-\alpha^{(t)}
\left ({\Cmat}_{\tilde{\xvec}\tilde{\xvec}}^{-1}  + \tilde{\Gmat}^{T}(\Lmat,\Vmat\hat{\tilde{\xvec}}^{(t)}){\Cmat}_{\tilde\wvec\tilde\wvec}^{-1}\tilde{\Gmat}(\Lmat,\Vmat\hat{\tilde{\xvec}}^{(t)})\right)^{-1}
\nonumber\\\times
 \left({\Cmat}_{\tilde{\xvec}\tilde{\xvec}}^{-1} (\hat{\tilde{\xvec}}^{(t)}- \tilde{\muvec}_\xvec)\hspace{4.5cm} \right.
\nonumber\\
\left.- \tilde\Gmat^{T}(\Lmat,\Vmat\hat{\tilde\xvec}^{(t)}){\Cmat}_{\tilde\wvec\tilde\wvec}^{-1}(\tilde\yvec-\tilde\gvec(\Lmat,\Vmat\hat{\tilde\xvec}^{(t)}))\right),\hspace{1.45cm}
\eeqna
where
\be
\label{GGG_tilde}
\tilde\Gmat(\Lmat,\Vmat\tilde\xvec)\define \nabla_{\tilde{\xvec}} \tilde\gvec(\Lmat,\Vmat\tilde{\xvec}) = \Vmat^{T}\Gmat(\Lmat,\xvec)\Vmat.
\ee
It can be verified that by multiplying \eqref{iteration_Newton2} by $\Vmat^T$ from the left, we obtain the iteration in the graph-frequency domain in \eqref{iteration_Newton2_GSP}.
In addition, it is known that the convergence of the Gauss-Newton method is invariant under affine transformations of the domain \cite{deuflhard1979affine}.
An advantage of the update in \eqref{iteration_Newton2_GSP} compared with \eqref{iteration_Newton2}  appears for cases where  $\xvec$ is a
graph bandlimited signal with a cutoff graph frequency $N_s$. In this case,   we substitute $[\hat{\tilde{\xvec}}^{(t)}]_n=0$, $\forall n>N_s$ and  update only 
the first $N_s$ elements of $\hat{\tilde{\xvec}}^{(t+1)}$ in
\eqref{iteration_Newton2_GSP} at each step.
In terms of computational complexity, implementing the MAP estimator by  \eqref{iteration_Newton2} and \eqref{iteration_Newton2_GSP} requires the computation of the inverse of $N\times N$ matrices $ {\Cmat}_{\xvec\xvec}^{-1}  + \Gmat^{T}(\Lmat,\hat{\xvec}^{(t)}){\Cmat}_{\wvec\wvec}^{-1}\Gmat(\Lmat,\hat{\xvec}^{(t)})$ in \eqref{iteration_Newton2},
or ${\Cmat}_{\tilde{\xvec}\tilde{\xvec}}^{-1}  + \tilde{\Gmat}^{T}(\Lmat,\Vmat\hat{\tilde{\xvec}}^{(t)}){\Cmat}_{\tilde\wvec\tilde\wvec}^{-1}\tilde{\Gmat}(\Lmat,\Vmat\hat{\tilde{\xvec}}^{(t)})$ in \eqref{iteration_Newton2_GSP} in addition to multiplications of $N\times N$ matrices
at each iteration, which leads to high computational complexity.

The MAP-estimator in the graph-frequency domain algorithm is summarized in Algorithm \ref{algorithm_MAP}.
\begin{algorithm}[hbt]
\SetAlgoLined
\KwInput{ 
	\begin{itemize}
		\item Function $\tilde\gvec(\Lmat,\Vmat\tilde\xvec)$
		\item Mean and covariance matrices $\tilde\muvec_{{\xvec}}$, ${\Cmat}_{{\tilde\xvec}{\tilde\xvec}}$, and ${\Cmat}_{\tilde\wvec\tilde\wvec}$
		\item Initial step size $\alpha_{0}$, $\gamma$ and $\Delta$
		\item Tolerance $\delta$ 
	\end{itemize}}
	{\textbf{Algorithm Steps:}}
	\begin{enumerate}
		\setlength\itemsep{0.2em}
		\item Initialization: ${\hat{\tilde\xvec}}^{(0)}$
		\item Compute: $\tilde\gvec(\Lmat,\Vmat\hat{\tilde\xvec}^{(t)})$ and $\tilde\Gmat(\Lmat,\Vmat\hat{\tilde\xvec}^{(t)})$ from \eqref{GGG_tilde}
		\item
		Choose step size $\alpha^{(t)}$ such that \eqref{line_search1} holds using Algorithm \ref{line_search_stages} with the objective function $Q$ from \eqref{Q_def_GSP}
		\item 	Compute
		${{\xvec}}^{(t+1)}$ from \eqref{iteration_Newton2_GSP} with step size $\alpha^{(t)}$
		\item Stopping condition: $||\hat{\tilde\xvec}^{(t+1)}-\hat{\tilde\xvec}^{(t)}||<\delta$ 
	\end{enumerate}
\KwOutput{MAP estimator:
	$\hat{\xvec}  = \Vmat{\tilde{\xvec}}^{(t+1)}$}
	\caption{MAP estimator in the graph-frequency domain by Gauss-Newton method}
	\label{algorithm_MAP}
\end{algorithm}

Due to the nonconvexity of $\gvec(\Lmat,\xvec)$ and the quadratic loss function, the Gauss-Newton method is sensitive to initialization and may diverge. 
These challenges inhibit its use for real-time estimation in large-scale networks. 
Moreover, these methods do not utilize the information about the underlying graph structure, e.g. for the initialization. 
Thus, they are less robust to changes and misspecification in the graph topology, e.g. in the graph connectivity. 
Finally, the estimators in \eqref{iteration_Newton2} and \eqref{iteration_Newton2_GSP} ignore the GSP properties and do not exploit additional information on the graph signal, such as smoothness or graph-bandlimitness, that could improve estimation performance.




\section{eGFD-MAP and GSP-MAP estimators}\label{eGFD_MAP_and_GSP_MAP}
In this section, we propose three new estimators that integrate the graph structure: the eGFD-MAP estimator in Subsection \ref{eGFD-MAP_EST}, and the sGSP-MAP and GSP-MAP estimators in Subsection \ref{MAP_GSP_section}.
In Subsection \ref{Disscussion2} we compare the estimators, focusing on their associated objective functions.
In Subsection \ref{separately_model_section}, we present special cases and discuss the conditions under which the proposed eGFD-MAP, sGSP-MAP, and GSP-MAP estimators coincide with the MAP estimator. Finally, in Subsection \ref{computational_complexity} we compare the computational complexity of the different estimators.

\subsection{eGFD-MAP estimator}\label{eGFD-MAP_EST}
In the nonlinear case, numerical optimization methods are used to minimize \eqref{Q_def}   or  \eqref{Q_def_GSP}.
 However, when $N$ is large, the minimization problem is high-dimensional, with high computational complexity and memory demands. 
In order to reduce the complexity and
accelerate the convergence rate of the iterative optimization algorithms, we
 propose here the eGFD-MAP estimator, which is comprised of two steps.
 In the first step, we replace the MAP objective function from \eqref{Q_def_GSP} by the following objective function
  in the graph-frequency domain:
\beqna
\label{Q_def_GSP_incremental}
Q_{freq}^{(d)}(\tilde{\xvec})\define
\frac{1}{2}(\tilde{\xvec}-\tilde{\muvec}_\xvec)^{T}{\Dmat}_{\tilde{\xvec}\tilde{\xvec}}^{(\text{inv})}(\tilde{\xvec}-\tilde{\muvec}_\xvec)\hspace{2.25cm}\nonumber\\ + \frac{1}{2}(\tilde{\yvec}-\tilde\gvec(\Lmat,\Vmat\tilde\xvec))^{T}{\Dmat}_{\tilde\wvec\tilde\wvec}^{(\text{inv})}(\tilde\yvec-\tilde\gvec(\Lmat,\Vmat\tilde\xvec))\nonumber\\
=\frac{1}{2}\sum_{n=1}^N
(\tilde{x}_n-[\tilde{\mu}_\xvec]_n)^2[{\Cmat}_{\tilde{\xvec}\tilde{\xvec}}^{-1}]_{n,n}\hspace{2.1cm}\nonumber\\ + \frac{1}{2}\sum_{n=1}^N(\tilde{y}_n-
[\tilde{\gvec}(\Lmat,\Vmat\tilde{\xvec})]_n
)^2[{\Cmat}_{\tilde\wvec\tilde\wvec}^{-1}]_{n,n},\hspace{0.5cm}
\eeqna
where 
\be
\label{inv_Dx_define}
\Dmat_{\tilde{\xvec}\tilde{\xvec}}^{(\text{inv})}\define {\text{ddiag}}({\Cmat}_{\tilde{\xvec}\tilde{\xvec}}^{-1})
\ee
and 
\be 
\label{inv_Dw_define}
\Dmat_{\tilde{\wvec}\tilde{\wvec}}^{(\text{inv})}\define {\text{ddiag}}({\Cmat}_{\tilde{\wvec}\tilde{\wvec}}^{-1}).
\ee

Similar to the derivations of
\eqref{iteration_Newton2} and \eqref{iteration_Newton2_GSP}, the approximated linearized objective function of $Q_{freq}^{(d)}(\tilde\xvec)$ is given by
\beqna
\label{Q_def_GSP_lin_incremental}
Q_{freq-lin}^{(d)}(\tilde{\xvec})\define
\frac{1}{2}(\tilde{\xvec}-\tilde{\muvec}_\xvec)^{T}{\Dmat}_{\tilde{\xvec}\tilde{\xvec}}^{(\text{inv})}(\tilde{\xvec}-\tilde{\muvec}_\xvec)\hspace{1.6cm}\nonumber\\ + \frac{1}{2}(\tilde{\yvec}-\tilde\gvec(\Lmat,\Vmat\hat{{\tilde\xvec}}^{(t)}) - \tilde\Gmat(\Lmat,\Vmat\hat{{\tilde\xvec}}^{(t)})(\tilde\xvec-\hat{{\tilde\xvec}}^{(t)}))^{T}{\Dmat}_{\tilde\wvec\tilde\wvec}^{(\text{inv})}\nonumber\\\times(\tilde{\yvec}-\tilde\gvec(\Lmat,\Vmat\hat{{\tilde\xvec}}^{(t)}) - \tilde\Gmat(\Lmat,\Vmat\hat{{\tilde\xvec}}^{(t)})(\tilde\xvec-\hat{{\tilde\xvec}}^{(t)})).
\eeqna
 The minimization of $Q_{freq-lin}^{(d)}(\tilde\xvec)$ w.r.t. $\tilde\xvec$ by the Gauss-Newton method results in the following update equation:
\beqna\label{27_old25_before_approx}
\hat{\tilde{\xvec}}^{(t+1)}
=\hat{\tilde{\xvec}}^{(t)}\hspace{6.25cm}\nonumber\\-\alpha^{(t)}\bigg({\Dmat}_{\tilde{\xvec}\tilde{\xvec}}^{(\text{inv})} + \tilde{\Gmat}(\Lmat,\Vmat\hat{\tilde{\xvec}}^{(t)})^{T}{\Dmat}_{\tilde{\wvec}\tilde{\wvec}}^{(\text{inv})}\tilde{\Gmat}(\Lmat,\Vmat\hat{\tilde{\xvec}}^{(t)})\bigg)^{-1}\nonumber\\\times\bigg(\Dmat_{\tilde{\xvec}\tilde{\xvec}}^{(\text{inv})}(\hat{\tilde{\xvec}}^{(t)}- \tilde{\muvec}_\xvec)\hspace{4.75cm}
\nonumber\\
-\tilde{\Gmat}(\Lmat,\Vmat\hat{\tilde{\xvec}}^{(t)})^{T}\Dmat_{\tilde\wvec\tilde\wvec}^{(\text{inv})}(\tilde\yvec-\tilde\gvec(\Lmat,\Vmat\hat{\tilde\xvec}^{(t)}))\bigg).\hspace{1.5cm}
\eeqna


In Subsection \ref{separately_model_section}, we present the orthogonal-graph-frequencies case, in which the estimator from \eqref{27_old25_before_approx} that minimizes the objective function in \eqref{Q_def_GSP_lin_incremental} is separable in the graph-frequency domain and can be implemented with per-coordinate iterations. Thus, in this case, the estimator from \eqref{27_old25_before_approx} has a lower computational complexity than the MAP estimator. 
However,  a major problem in the general case is that
the matrix  $\tilde\Gmat(\Lmat,\Vmat\hat{\tilde\xvec}^{(t)})$ in \eqref{27_old25_before_approx}  is a full matrix that changes at each iteration.
To bypass this hurdle,  we use the Gauss-Newton iteration in \eqref{27_old25_before_approx} with an additional step of neglecting the 
off-diagonal elements of 
$\tilde\Gmat(\Lmat,\Vmat\hat{\tilde\xvec}^{(t)})$.
This approach results in a separable form of the estimator in the graph-frequency domain for the general non-orthogonal case.


In the second step, at each iteration of \eqref{27_old25_before_approx} we neglect the non-diagonal elements of the Jacobian matrix, $\tilde{\Gmat}(\Lmat,\Vmat\hat{\tilde{\xvec}}^{(t)})$, which involves the mixed-derivatives of 
 $\tilde\gvec(\Lmat,\Vmat\tilde\xvec)$.  This results in the following iteration:
\beqna\label{27_old25}
\hat{\tilde{\xvec}}^{(t+1)}
=\hat{\tilde{\xvec}}^{(t)}-\alpha^{(t)}\bigg({\Dmat}_{\tilde{\xvec}\tilde{\xvec}}^{(\text{inv})} + {\Dmat}_{\tilde{\wvec}\tilde{\wvec}}^{(\text{inv})}\bar\Dmat_{\tilde{\Gmat}}(\hat{\tilde{\xvec}}^{(t)})^2\bigg)^{-1}\nonumber\\\times\bigg(\Dmat_{\tilde{\xvec}\tilde{\xvec}}^{(\text{inv})}(\hat{\tilde{\xvec}}^{(t)}- \tilde{\muvec}_\xvec)
\hspace{3.25cm}
\nonumber\\
-\Dmat_{\tilde\wvec\tilde\wvec}^{(\text{inv})}\bar\Dmat_{\tilde{\Gmat}}(\hat{\tilde{\xvec}}^{(t)})(\tilde\yvec-\tilde\gvec(\Lmat,\Vmat\hat{\tilde\xvec}^{(t)}))\bigg),\hspace{0.6cm} \eeqna
where \[\bar\Dmat_{\tilde{\Gmat}}(\tilde\xvec)=\text{ddiag}(\tilde{\Gmat}(\Lmat,\Vmat\tilde\xvec)).\]
It can be seen that
 the estimator in \eqref{27_old25}
is the estimator that minimizes the objective function from \eqref{Q_def_GSP_lin_incremental} after replacing $\tilde{\Gmat}(\Lmat,\Vmat\tilde\xvec)$ with $\bar\Dmat_{\tilde{\Gmat}}(\tilde\xvec)$, i.e. the following objective function:
\beqna
\label{Q_def_GSP_lin_incremental_approx}
Q_{freq-lin}^{(d, approx)}(\tilde{\xvec})\define
\frac{1}{2}(\tilde{\xvec}-\tilde{\muvec}_\xvec)^{T}{\Dmat}_{\tilde{\xvec}\tilde{\xvec}}^{(\text{inv})}(\tilde{\xvec}-\tilde{\muvec}_\xvec)\hspace{1.6cm}\nonumber\\ + \frac{1}{2}(\tilde{\yvec}-\tilde\gvec(\Lmat,\Vmat\hat{{\tilde\xvec}}^{(t)}) - \bar\Dmat_{\tilde{\Gmat}}(\hat{\tilde{\xvec}}^{(t)})(\tilde\xvec-\hat{{\tilde\xvec}}^{(t)}))^{T}{\Dmat}_{\tilde\wvec\tilde\wvec}^{(\text{inv})}\nonumber\\\times(\tilde{\yvec}-\tilde\gvec(\Lmat,\Vmat\hat{{\tilde\xvec}}^{(t)}) - \bar\Dmat_{\tilde{\Gmat}}(\hat{\tilde{\xvec}}^{(t)})(\tilde\xvec-\hat{{\tilde\xvec}}^{(t)})).\hspace{1.25cm}
\eeqna
The estimator in \eqref{27_old25} is denoted as the eGFD-MAP estimator.

Gradient-based iterative optimization
algorithms typically converge slowly for large-scale or poorly-conditioned inverse problems, such as the problem of finding the MAP estimator. 
Preconditioning is a technique that significantly improves the convergence rate by transforming large matrices with alternative matrices that are easy to invert 
\cite{Preconditioning_Clinthorne_1993,hanke1992preconditioned}. A typically employed strategy is to use diagonal approximation as a preconditioning matrix. In this context, the eGFD-MAP estimator can be viewed as diagonal preconditioning applied to the Gauss-Newton iteration of the MAP estimator in the graph frequency domain.

The advantages of the eGFD-MAP estimator compared with the conventional MAP iterative methods from Subsection \ref{iterative_subsection} are that: 1) there is no need to calculate the off-diagonal elements of the Jacobian matrix $\tilde\Gmat(\Lmat,\Vmat\tilde\xvec)$ at each iteration; and 2) there is no need to perform matrix inversion per iteration.
Indeed, in order to compute $\Dmat_{\tilde{\xvec}\tilde{\xvec}}^{(\text{inv})}$ and $\Dmat_{\tilde{\wvec}\tilde{\wvec}}^{(\text{inv})}$ it is necessary to compute the inverse of the prior covariance matrices ${\Cmat}_{\tilde{\xvec}\tilde{\xvec}}$ and ${\Cmat}_{\tilde{\wvec}\tilde{\wvec}}$. However, this calculation does not change at each iteration, and thus can be done offline where ${\Cmat}_{\tilde{\xvec}\tilde{\xvec}}$ and ${\Cmat}_{\tilde{\wvec}\tilde{\wvec}}$ are assumed to be known.
In particular, assuming that the covariance matrices and the Jacobian matrix are given, performing a single iteration of the MAP update equation according to \eqref{iteration_Newton2} or \eqref{iteration_Newton2_GSP} requires $\mathcal{O}(N^{3})$ calculations;  this is since, in general, it is necessary to perform matrix inversion, while the update step of the eGFD-MAP estimator from \eqref{27_old25} requires only $\mathcal{O}(N)$ calculations.  
Due to these two advantages, and the fact that \eqref{27_old25} can be implemented in a component-wise  fashion,
the eGFD-MAP estimator in \eqref{27_old25}
can result in a significant overall speedup of computation and can be used even when the size of the network is large; this is in contrast with the MAP estimator, which becomes intractable for large networks. 
In addition, in cases where the matrices $ {\Cmat}_{\xvec\xvec}^{-1}  + \Gmat^{T}(\Lmat,\hat{\xvec}^{(t)}){\Cmat}_{\wvec\wvec}^{-1}\Gmat(\Lmat,\hat{\xvec}^{(t)})$ in \eqref{iteration_Newton2}
or ${\Cmat}_{\tilde{\xvec}\tilde{\xvec}}^{-1}  + \tilde{\Gmat}^{T}(\Lmat,\Vmat\hat{\tilde{\xvec}}^{(t)}){\Cmat}_{\tilde\wvec\tilde\wvec}^{-1}\tilde{\Gmat}(\Lmat,\Vmat\hat{\tilde{\xvec}}^{(t)})$ in \eqref{iteration_Newton2_GSP}, are ill-conditioned (as may happen when the sample covariance matrices are used instead of the true covariances), the calculation of their inverse is prone to large numerical errors. This can badly affect the performance of the MAP estimator, in contrast to the performance of the eGFD-MAP estimator that does not require matrix inversion.
The eGFD-MAP-estimator algorithm is summarized in Algorithm \ref{algorithm_general}.
\vspace{-0.25cm}
\begin{algorithm}[hbt]
\SetAlgoLined
	\textbf{Input:}	
	\begin{itemize}
		\item Function $\tilde\gvec(\Lmat,\Vmat\tilde\xvec)$
		\item Laplacian matrix $\Lmat$
		\item Initial step size $\alpha^{(0)}$, $\gamma$ and $\Delta$
		\item Mean and the diagonal entries of the inverse covariance matrices: $\tilde\muvec_{\xvec}$, ${\Dmat}_{\tilde{\xvec}\tilde{\xvec}}^{(\text{inv})}$, and ${\Dmat}_{\tilde{\wvec}\tilde{\wvec}}^{(\text{inv})}$
		\item Tolerance $\delta$
	\end{itemize}
	{\textbf{Algorithm Steps:}}
	\begin{enumerate}
		\setlength\itemsep{0.2em}
		\item Initialization: ${\hat{\tilde{\xvec}}}^{(0)}$
		\item Compute 
		$\tilde\gvec(\Lmat,\Vmat\hat{\tilde\xvec}^{(t)})$ and 
		$\bar\Dmat_{\tilde{\Gmat}}(\hat{\tilde{\xvec}}^{(t)})$ from \eqref{GGG_tilde}
		\item Choose step size $\alpha^{(t)}$ according to Algorithm \ref{line_search_stages} with the objective function $Q_{freq}^{(d)}$ from \eqref{Q_def_GSP_incremental}
		\item\label{Step4} Apply the iteration step according to
		\eqref{27_old25}
		\item Stopping condition: $||{\hat{\tilde\xvec}}^{(t+1)}-{\hat{\tilde\xvec}}^{(t)}||<\delta$ 
	\end{enumerate}
	\vspace{1mm}
	\KwOutput{eGFD-MAP estimator
	$\hat{\xvec}  = \Vmat{\hat{\tilde{\xvec}}}^{(t+1)}$}
	\caption{eGFD-MAP estimator by Gauss-Newton method}
	\label{algorithm_general}

\end{algorithm}

 \vspace{-0.5cm}
\subsection{sGSP-MAP and GSP-MAP estimators}
\label{MAP_GSP_section}
While the eGFD-MAP estimator has a lower computational cost, this advantage comes at the expense of neglecting the off-diagonal elements of ${\Cmat}_{\tilde{\xvec}\tilde{\xvec}}^{-1}$, ${\Cmat}_{\tilde{\wvec}\tilde{\wvec}}^{-1}$, and $
\tilde{\Gmat}(\Lmat,\Vmat\tilde\xvec)$.
Consequently, it is expected to have good performance 
in cases where the elements of the signals at the different graph frequencies are (almost) uncorrelated. 
  However, this may not be the case, and in this subsection we propose an alternative approach that does not assume a lack of correlation between the elements of the signals in the graph-frequency domain, but still utilizes the GSP properties. 
As shown in Subsection \ref{iterative_subsection}, the update equation of the MAP estimator is obtained by solving a
linearized WLS problem at each iteration.
As a result,   the update equations under the three objective functions in \eqref{iteration_Newton2}, \eqref{iteration_Newton2_GSP}, and \eqref{27_old25} are all linear functions of $\yvec-\gvec(\Lmat,\hat{\xvec}^{(t)})$ and of $\hat{\xvec}^{(t)}- \muvec_\xvec$.
Based on this representation, 
in the following we remain
with the linearized WLS problem,
but constrain the estimator at each iteration $t$ to be the output of two graph filters. Based on this representation, we now propose the sGSP-MAP and GSP-MAP estimators.

In the following,  
we consider that at the $t$th iteration, the estimator has the form:
\beqna \label{opt_estimator}
	\hat{\xvec}^{(t+1)} =\hat{\xvec}^{(t)}
	+ f_1(\Lmat,\hat{\xvec}^{(t)})(\hat{{\xvec}}^{(t)}- {\muvec}_\xvec)\nonumber\\
	+f_2(\Lmat,\hat{\xvec}^{(t)}) (\yvec-\gvec(\Lmat,\hat{\xvec}^{(t)}))  ,
\eeqna
where
 $f_i(\cdot,\cdot)$, $i=1,2$, are graph filters as defined in \eqref{laplacian_graph_filter}.
By left-multiplying \eqref{opt_estimator}
 by $\Vmat^T$, we obtain that the estimator from \eqref{opt_estimator} can be written in the graph frequency domain  as 
\beqna \label{opt_estimator_graph}
	\hat{\tilde{\xvec}}^{(t+1)}  
	= \hat{\tilde{\xvec}}^{(t)}+  f_1(\Lambdamat,\Vmat\hat{\tilde{\xvec}}^{(t)} )(\hat{\tilde{\xvec}}^{(t)}- \tilde{\muvec}_\xvec)\nonumber\\
	+f_2(\Lambdamat,\Vmat\hat{\tilde{\xvec}}^{(t)} )(\tilde\yvec-\tilde\gvec(\Lmat,\Vmat\hat{\tilde\xvec}^{(t)})),
\eeqna
where $\Vmat$ and $\Lambdamat$ are 
defined in \eqref{SVD_new_eq}.

In this form of estimators, the terms $\hat{\tilde{\xvec}}^{(t)}- \tilde{\muvec}_\xvec$ and $\tilde\yvec-\tilde\gvec(\Lmat,\Vmat\hat{\tilde\xvec}^{(t)})$ are multiplied by {\em{diagonal}} matrices that represent the graph filters.
It can be seen that in the general case,
 the estimators in \eqref{iteration_Newton2} and  \eqref{iteration_Newton2_GSP} cannot be written in the form of \eqref{opt_estimator} or \eqref{opt_estimator_graph}. 
 In contrast, 
the iteration  of the eGFD-MAP estimator from \eqref{27_old25} can be written as the output of graph filters, as described in
\eqref{opt_estimator_graph}, where  $f_1(\cdot,\cdot)$ and $f_2(\cdot,\cdot)$ are the following graph filters:
\beqna
\label{f1_def}
f_1^{\text{eGFD-MAP}}(\Lambdamat,\Vmat\hat{\tilde{\xvec}}^{(t)})\hspace{4cm}\nonumber\\=-
 \alpha^{(t)}\bigg({\Dmat}_{\tilde{\xvec}\tilde{\xvec}}^{(\text{inv})} + {\Dmat}_{\tilde{\wvec}\tilde{\wvec}}^{(\text{inv})}
\bar\Dmat_{\tilde{\Gmat}}^2(\hat{\tilde{\xvec}}^{(t)})\bigg)^{-1} 
\Dmat_{\tilde{\xvec}\tilde{\xvec}}^{(\text{inv})} 
\eeqna
and
\beqna
\label{f2_def}
f_2^{\text{eGFD-MAP}}(\Lambdamat,\Vmat\hat{\tilde{\xvec}}^{(t)})=
\hspace{4cm}\nonumber\\
 \alpha^{(t)}\bigg({\Dmat}_{\tilde{\xvec}\tilde{\xvec}}^{(\text{inv})} + {\Dmat}_{\tilde{\wvec}\tilde{\wvec}}^{(\text{inv})}
\bar\Dmat_{\tilde{\Gmat}}^2(\hat{\tilde{\xvec}}^{(t)})\bigg)^{-1}
 \Dmat_{\tilde\wvec\tilde\wvec}^{(\text{inv})}
\bar\Dmat_{\tilde{\Gmat}}(\hat{\tilde{\xvec}}^{(t)}).
\eeqna

In the following, our goal is to choose the graph filters
$f_1(\cdot,\cdot)$ and $f_2(\cdot,\cdot)$ in the updated equation in \eqref{opt_estimator} (or in \eqref{opt_estimator_graph}) in an optimal way, in the sense that the expected objective function $ Q_{lin}$ from \eqref{Q_def_approx} is minimized for the general case.
 It should be noted that the graph filters can be a function of the previous-iteration estimator, $\hat{{\xvec}}^{(t)}$.
 The following theorem describes the optimal graph filters, in the sense of minimizing  $ Q_{lin}$.
\textcolor{black}{
To this end, by considering the model in \eqref{Model}, i.e. $\yvec = \gvec(\Lmat,\xvec) + \wvec$ with uncorrelated signal and noise,  we define  the sample covariance matrices in the graph-frequency domain that are based on the single samples $\hat{\tilde{\xvec}}^{(t)}$ and $\tilde{\yvec}$ as follows:
\be
\label{Smat_x}
\Smat_{\hat{\tilde{\xvec}}\hat{\tilde{\xvec}}}^{(t)}\define 
(\hat{\tilde{\xvec}}^{(t)}-\tilde\muvec_{\xvec})(\hat{\tilde{\xvec}}^{(t)}-\tilde\muvec_{\xvec})^{T},
\ee
\be
\label{Smat_y}
\Smat_{\tilde{\wvec}\tilde{\wvec}}^{(t)}\define 
({\tilde{\yvec}}-\tilde\gvec(\Lmat, \Vmat\hat{\tilde{\xvec}}^{(t)}))({\tilde{\yvec}}-\tilde\gvec(\Lmat, \Vmat\hat{\tilde{\xvec}}^{(t)})^{T},
\ee
and
\be
\label{Smat_yx}
\Smat_{\tilde{\wvec}\hat{\tilde{\xvec}}}{^{(t)}}
\define ({\tilde{\yvec}}-\tilde\gvec(\Lmat, \Vmat\hat{\tilde{\xvec}}^{(t)}))(\hat{\tilde{\xvec}}^{(t)}-\tilde\muvec_{\xvec})^{T}.
\ee
\begin{Theorem}
\label{claim2}
The graph filters that minimize the objective function
$Q_{lin}(\hat{\xvec}^{(t+1)},\hat{{\xvec}}^{(t)})$ from \eqref{Q_def_approx}
over the subset of GSP estimators defined in \eqref{opt_estimator} under the assumption that $\Smat_{\tilde{\wvec}\hat{\tilde{\xvec}}}^{(t)}$ is negligible are
\beqna
\label{filter1_noE}
f_1^{\text{sGSP-MAP}}(\Lambdamat,\hat{\xvec}^{(t)} ) \nonumber = -
{\emph{diag}}\bigg(\Big(\Smat_{\hat{\tilde{\xvec}}\hat{\tilde{\xvec}}}^{(t)}
\circ ({\Cmat}_{\tilde\xvec\tilde\xvec}^{-1}\hspace{1.5cm} \\+ \tilde\Gmat^T(\Lmat,\Vmat\hat{\tilde{\xvec}}^{(t)}) 
{\Cmat}_{\tilde\wvec\tilde\wvec}^{-1}
\tilde\Gmat(\Lmat,\Vmat\hat{\tilde{\xvec}}^{(t)}))\Big)^{-1}\nonumber\\
\times   {\emph{diag}}\big(\Smat_{\hat{\tilde{\xvec}}\hat{\tilde{\xvec}}}^{(t)}{\Cmat}_{\tilde\xvec\tilde\xvec}^{-1} \big)\bigg)\hspace{2.8cm} 
\eeqna
and
\beqna
\label{filter2_noE}
f_2^{\text{sGSP-MAP}}(\Lambdamat,\hat{\xvec}^{(t)} )=
{\emph{diag}}\bigg(\Big(\Smat_{\tilde{\wvec}\tilde{\wvec}}^{(t)}\nonumber\circ \big({\Cmat}_{\tilde\xvec\tilde\xvec}^{-1} \hspace{1.5cm} \\ \nonumber + \tilde\Gmat^T(\Lmat,\Vmat\hat{\tilde{\xvec}}^{(t)}) 
{\Cmat}_{\tilde\wvec\tilde\wvec}^{-1}
\tilde\Gmat(\Lmat,\Vmat\hat{\tilde{\xvec}}^{(t)})\Big)^{-1}
\nonumber\\
\times
{ \emph{diag}}\big(  \Smat_{\tilde{\wvec}\tilde{\wvec}}^{(t)}{\Cmat}_{\tilde\wvec\tilde\wvec}^{-1}
\tilde\Gmat(\Lmat,\Vmat\hat{\tilde{\xvec}}^{(t)})\big) \bigg) \hspace{0.5cm} 
\eeqna
at the $(t+1)$th iteration.
\end{Theorem}
\begin{IEEEproof}
The proof appears in Appendix \ref{appendix_derivation_graph_filters_noE}.
\end{IEEEproof}
We denote the estimator that is obtained by substituting the graph filters from Theorem \ref{claim2} in \eqref{opt_estimator} by sGSP-MAP estimator since it is based on the sample covariance matrices from \eqref{Smat_x}-\eqref{Smat_yx}. 
The  rationale behind neglecting  the sample cross-covariance matrix, i.e. approximating
$\Smat_{\tilde{\wvec}\hat{\tilde{\xvec}}}^{(t)}\approx \zerovec$, is that under the model assumptions,  the true covariance matrices satisfy $\Cmat_{\wvec\xvec}=\Cmat_{\tilde\wvec\tilde\xvec}=\zerovec$.  Without this approximation, the algorithm is very sensitive and unreliable.
In Section \ref{simulation}, we demonstrate that the performance of the sGSP-MAP estimator is close to the MAP estimator's performance, and hence, the approximation is justified. 
The sGSP-MAP estimator uses the singular one-sample covariance matrices $\Smat_{\hat{\tilde{\xvec}}\hat{\tilde{\xvec}}}^{(t)}$ and $\Smat_{{\tilde{\wvec}}{\tilde{\wvec}}}^{(t)}$. 
While these matrices provide an unbiased estimation of the covariance matrices, they have a large variance and are  ill-posed
\cite{Wiesel_Eldar2010}. As a result,
 the calculation of the graph filters from \eqref{filter1_noE} and \eqref{filter2_noE}, which is based on inverting matrices with a low condition number, is unstable.
 In order to alleviate the numerical instability,  the sGSP-MAP estimator is implemented with an ad-hoc diagonal loading approach \cite{Stoica_Wang,Carlson_1988}. 
\\
\indent In order to further increase the overall stability of the estimation method, we change the objective function, $Q_{lin}$,  to a smoother one, obtained by the expected value. Taking the expectation makes the resulting objective function more amenable to iterative techniques such as the Gauss-Newton method. 
The following theorem describes the GSP-MAP graph filters that minimize the expected objective function.
}
\begin{Theorem}
\label{claim1}
The graph filters that minimize the expected objective function, \be
\label{obj1}
{\rm{E}}[Q_{lin}(\hat{\xvec}^{(t+1)},\hat{{\xvec}}^{(t)})|{\hat{{\xvec}}^{(t)}=\xvec}],
\ee 
over the subset of GSP estimators in \eqref{opt_estimator}, 
under
 the approximation that $\tilde\Gmat(\Lmat,\Vmat\hat{\tilde{\xvec}}^{(t)})$ is  a deterministic  matrix\footnote{The rationale behind this approximation is similar to that behind the first-order approximation in the Gauss-Newton method \cite{Bell_Cathey1993}.},
are
\beqna
\label{filter1}
f_1^{\text{GSP-MAP}}(\Lambdamat,\hat{\xvec}^{(t)} )=-
{\emph{diag}}\bigg(\Big({\Cmat}_{\tilde\xvec\tilde\xvec}
\circ ({\Cmat}_{\tilde\xvec\tilde\xvec}^{-1}
\nonumber\\+\tilde\Gmat^T(\Lmat,\Vmat\hat{\tilde{\xvec}}^{(t)}) {\Cmat}_{\tilde\wvec\tilde\wvec}^{-1}\tilde\Gmat(\Lmat,\Vmat\hat{\tilde{\xvec}}^{(t)}))\Big)^{-1}
\onevec\bigg)
\eeqna
and
\beqna
\label{filter2}
f_2^{\text{GSP-MAP}}(\Lambdamat,\hat{\xvec}^{(t)} )=
{\emph{diag}}\bigg(({\Cmat}_{\tilde\wvec\tilde\wvec} \circ 
({\Cmat}_{\tilde\xvec\tilde\xvec}^{-1}
\nonumber\\+
\tilde\Gmat^T(\Lmat,\Vmat\hat{\tilde{\xvec}}^{(t)}) {\Cmat}_{\tilde\wvec\tilde\wvec}^{-1}\tilde\Gmat(\Lmat,\Vmat\hat{\tilde{\xvec}}^{(t)})))^{-1}
\nonumber\\
\times{\emph{diag}}(\tilde\Gmat(\Lmat,\Vmat\hat{\tilde{\xvec}}^{(t)}))\bigg)
\eeqna
at the $(t+1)$th iteration.
\end{Theorem}
We denote the estimator
 that is obtained by substituting the graph filters 
from Theorem \ref{claim1}
 in \eqref{opt_estimator} 
by GSP-MAP estimator.
\begin{IEEEproof}
The proof appears in Appendix \ref{appendix_derivation_graph_filters}.
\end{IEEEproof}

Similar to the  iterative  algorithms from Section \ref{iterative_subsection}, in practice, we multiply the sGSP-MAP estimator and the graph filters  by the step size
$\alpha^{(t)}\in(0,1]$ that is computed by a backtracking line search. The sGSP-MAP and the GSP-MAP algorithms are summarized in Algorithm \ref{new_alg}.
\begin{algorithm}[hbt]
\SetAlgoLined
	\textbf{Input:}	
	\begin{itemize}
		\item Function $\tilde\gvec(\Lmat,\Vmat\tilde\xvec)$
		\item Laplacian matrix $\Lmat$
		\item Initial step size $\alpha^{(0)}$, $\gamma$ and $\Delta$
\item 	Mean and covariance matrices $\tilde\muvec_{{\xvec}}$, ${\Cmat}_{\tilde{\xvec}\tilde{\xvec}}$, and ${\Cmat}_{\tilde\wvec\tilde\wvec}$
\item Tolerance $\delta$ 
	\end{itemize}
	{\textbf{Algorithm Steps:}}
	\begin{enumerate}
		\setlength\itemsep{0.2em}
		\item Initialization: ${\hat{\tilde{\xvec}}}^{(0)}$
		\item Compute 
		$\tilde\gvec(\Lmat,\Vmat\hat{\tilde\xvec}^{(t)})$ and $\tilde\Gmat(\Lmat,\Vmat\hat{\tilde\xvec}^{(t)})$ from \eqref{GGG_tilde}
		\item {\textbf{Option A: sGSP-MAP estimator:}}\\
   Compute the graph filters: $f_1^*(\Lambdamat,\hat{\xvec}^{(t)} )=
f_1^{\text{sGSP-MAP}}(\Lambdamat,\hat{\xvec}^{(t)} ) $ from \eqref{filter1_noE}
		$f_2^*(\Lambdamat,\hat{\xvec}^{(t)} )=f_2^{\text{sGSP-MAP}}(\Lambdamat,\hat{\xvec}^{(t)} ) $
		from  \eqref{filter2_noE}\\
  with additional diagonal loading approaches.
  \\
  {\textbf{Option B: GSP-MAP estimator:}}\\
  Compute the graph filters:
  $f_1^*(\Lambdamat,\hat{\xvec}^{(t)} )=f_1^{\text{GSP-MAP}}(\Lambdamat,\hat{\xvec}^{(t)} )$
		from \eqref{filter1}
		$f_2^*(\Lambdamat,\hat{\xvec}^{(t)} )=f_2^{\text{GSP-MAP}}(\Lambdamat,\hat{\xvec}^{(t)} )$
		from \eqref{filter2} 
		\item Choose step size $\alpha^{(t)}$ according to Algorithm \ref{line_search_stages} with the objective function $Q_{freq}$.
		\item Apply iteration step in \eqref{opt_estimator_graph} with the graph filters:
		\beqna \label{opt_estimator_graph_filter}	\hspace{-0.5cm}\hat{\tilde{\xvec}}^{(t+1)}  
	= \hat{\tilde{\xvec}}^{(t)}+ \alpha^{(t)} f_1^*(\Lambdamat,\Vmat\hat{\tilde{\xvec}}^{(t)} )(\hat{\tilde{\xvec}}^{(t)}- \tilde{\muvec}_\xvec)\nonumber\\
	+\alpha^{(t)} f_2^*(\Lambdamat,\Vmat\hat{\tilde{\xvec}}^{(t)} )(\tilde\yvec-\tilde\gvec(\Lmat,\Vmat\hat{\tilde\xvec}^{(t)}))
\eeqna
		\item Stopping condition: $||{\hat{\tilde\xvec}}^{(t+1)}-{\hat{\tilde\xvec}}^{(t)}||<\delta$ 
	\end{enumerate}
	\vspace{1mm}
	\KwOutput{Graph-filtered GSP-MAP estimator
	$\hat{\xvec} \hspace{-0.12cm} =\hspace{-0.13cm} \Vmat{\hat{\tilde{\xvec}}}^{(t+1)}\hspace{-0.2cm}$}
	\caption{sGSP-MAP and GSP-MAP estimators by Gauss-Newton method}
	\label{new_alg}
\end{algorithm}

It can be seen that
replacing $\Smat_{\hat{\tilde{\xvec}}\hat{\tilde{\xvec}}}^{(t)}$  and $\Smat_{{\tilde{\wvec}}{\tilde{\wvec}}}^{(t)}$ 
with ${\Cmat}_{\tilde\xvec\tilde\xvec}$ 
and ${\Cmat}_{\tilde\wvec\tilde\wvec}$ in the graph filters from Theorem \ref{claim2}, \eqref{filter1_noE} and \eqref{filter2_noE},
results in the GSP-MAP optimal filters from \eqref{filter1} and \eqref{filter2}.
Thus, the GSP-MAP estimator can be obtained from the sGSP-MAP estimator by using the true covariance matrices instead of the one-sample covariance matrices from \eqref{Smat_x} and \eqref{Smat_y}.
Consequently, the sGSP-MAP estimator fits better to the data, since it uses the current samples to obtain the graph filters, and since it minimizes the same objective function as the MAP estimator. On the other hand, the GSP-MAP estimator is more stable and uses matrices with a higher condition number.
\subsection{Discussion}
\label{Disscussion2}
In this subsection, we investigate the differences in the objective functions of the different estimators, as
summarized in Table \ref{objective_table}. The right column indicates if this estimator can be represented in the GSP form of \eqref{opt_estimator}. 
\renewcommand{\arraystretch}{2.5}
\begin{table}[hbt!]\centering
\begin{tabular}{|c|l|c|}
\hline
\fontsize{9}{9}\textbf{Estimator} & \fontsize{9}{9}\textbf{Objective function}  & \fontsize{9}{9}\textbf{GSP}                                                                                                                \\ \hline
{MAP}       & $Q_{lin}$    from  \eqref{Q_def_approx}                                                                &\xmark                               \\ \hline
{eGFD-MAP}  & \makecell{$Q_{freq-lin}^{(d,approx)}$ from  \eqref{Q_def_GSP_lin_incremental_approx}, i.e. 
$Q_{lin}$    from  \eqref{Q_def_approx}\\ under diagonal approximations}&\cmark                                                      \\ \hline
{sGSP-MAP}   & $Q_{lin}$    from  \eqref{Q_def_approx}&\cmark \\ \hline
{GSP-MAP}   & ${\rm{E}}\left[\left.Q_{lin}(\hat{\xvec}^{(t+1)},\hat{{\xvec}}^{(t)})\right|{\hat{{\xvec}}^{(t)}=\xvec}\right] \text{  from  } \eqref{obj1}$&\cmark \\ \hline
\end{tabular}
\caption{Comparison of the different estimators. }
\label{objective_table}
\end{table}

As discussed after Algorithm \ref{algorithm_MAP},
the original non-convex objective function of the MAP estimator 
 in \eqref{Q_def_approx}  
is sensitive to the initialization,  may diverge, 
has high computational complexity, 
and does not utilize any information about the underlying graph structure.
The low-complexity eGFD-MAP estimator approximates the covariance and Jacobian matrices in the MAP objective function by diagonal matrices, which can be interpreted as applying preconditioning \cite{Preconditioning_Clinthorne_1993,hanke1992preconditioned}. In that way, it reduces the computational complexity of the estimation approach 
and 
increases its robustness to bad initialization. 
The following claim discusses the relations between the MAP estimator and the eGFD-MAP estimator.
\begin{Claim}
\label{Q_lin_Q_freq_coincide} 
The objective function from \eqref{Q_def_GSP_incremental}, $Q_{freq}^{(d)}(\tilde{\xvec})$,
coincides with the MAP objective function in the graph-frequency domain, 
$Q_{freq}(\tilde\xvec)$, from \eqref{Q_def_GSP} 
if  the elements of $\tilde{\xvec}$ and $\tilde{\wvec}$  in the graph-frequency domain are uncorrelated,
i.e. if 
1)  ${\Cmat}_{\tilde{\xvec}\tilde{\xvec}} = \text{ddiag}({\Cmat}_{\tilde{\xvec}\tilde{\xvec}})$; and
2) ${\Cmat}_{\tilde{\wvec}\tilde{\wvec}} = \text{ddiag}({\Cmat}_{\tilde{\wvec}\tilde{\wvec}})$.
\end{Claim}
\begin{proof}
    It can be seen that when Conditions 1) and 2) are satisfied, 
    $Q_{freq}^{(d)}(\tilde{\xvec})$ from \eqref{Q_def_GSP_incremental}
      is readily obtained 
 from 
$Q_{freq}(\tilde\xvec)$ 
  by replacing the  inverse of the full covariance matrices ${\Cmat}_{\tilde{\xvec}\tilde{\xvec}}$ and  ${\Cmat}_{\tilde{\wvec}\tilde{\wvec}}$ in \eqref{Q_def_GSP} with their
diagonal versions in \eqref{inv_Dx_define} and \eqref{inv_Dw_define}, respectively.
\end{proof}
Although Claim \ref{Q_lin_Q_freq_coincide} does not guarantee that the MAP and the eGFD-MAP estimators coincide (since the mixed derivatives of  $\tilde\gvec(\Lmat, \Vmat\tilde\xvec)$ may be non-negligible), in practice (see Section \ref{simulation}) the performance of the eGFD-MAP estimator is comparable to that of the MAP estimator under Conditions 1) and 2).
However, when the conditions of Claim \ref{Q_lin_Q_freq_coincide} are not satisfied,   
the eGFD-MAP estimator
might 
deviate from the 
MAP estimator. 
 In contrast, the sGSP-MAP and the GSP-MAP estimators do not impose any assumptions on the covariance matrices of $\xvec$ and $\wvec$, nor on the Jacobian $\tilde\Gmat(\Lmat,\Vmat\hat{\tilde\xvec}^{(t)})$. 
The GSP estimators assume the structure of the output of two graph filters, which can be interpreted as regularization that enforces using the graphical information for the iteration of the MAP estimator. Thus, the proposed estimators have the advantage that they can be designed by using parametric graph filter design \cite{sparse_paper} and can be used to obtain a distributed implementation that works locally over the graph 
  \cite{5982158}.
However, 
 the eGFD-MAP estimator may not fully utilize the GSP information, since it employs the diagonal structure. Moreover, 
the sGSP-MAP estimator is adversely affected by the use of the sample covariance matrices that are only based on a single observation. 
Thus, the GSP-MAP estimator, which minimizes the expectation of the objective function of the MAP estimator \eqref{obj1}, is a smoother version of the sGSP-MAP estimator that is more robust to poor initialization.
It is important to mention that applying the expectation to the objective function of the MAP estimator from \eqref{iteration_Newton2} or to the eGFD-MAP estimator from \eqref{Q_def_GSP_lin_incremental_approx} in a similar manner as performed in Theorem \ref{claim1}, results in the same MAP and eGFD-MAP estimators, and does not improve the robustness of these estimators. Thus, the expected objective function can only be used in developing the GSP-MAP estimator, i.e. with the MAP objective function, $Q_{lin}$, together with the specific graph-filter structure from \eqref{opt_estimator}.




\subsection{Orthogonal-graph-frequencies models} 
\label{separately_model_section}
In this subsection,  we present the special case of orthogonal graph frequencies.
We show that in this case, 
 the proposed estimators coincide with the MAP estimator. 
The orthogonal-graph-frequencies model is defined as follows.
\begin{definition}
\label{orthogonal_frequencies}
The nonlinear measurements function $\gvec(\Lmat,\xvec)$ is separable in the graph-frequency domain (``orthogonal graph frequencies") if it satisfies
    \begin{equation} \label{separately_Model}
        [\tilde{\gvec}(\Lmat,\Vmat\tilde\xvec)]_n  = [\tilde{\gvec}(\Lmat,\tilde{x}_n\vvec_n)]_n,~n=1,\ldots,N,
    \end{equation}
    where $\tilde{x}_n$ is the $n$th element of $\tilde{\xvec}$ and $\tilde{\gvec}(\Lmat,\Vmat\tilde\xvec)=\Vmat^T\gvec(\Lmat,\xvec)$.
\end{definition}
Definition \ref{orthogonal_frequencies} is satisfied, for example, if the associated Jacobian matrix, $\Gmat(\Lmat,\xvec)$, is diagonalized by the eigenvector matrix  of  $\Lmat$, $\Vmat$.

By substituting \eqref{separately_Model} in \eqref{Q_def_GSP_incremental}, we obtain that in this case
\beqna
\label{Q_def_GSP_incremental2}
Q_{freq}^{(d)}(\tilde{\xvec})
=\frac{1}{2}\sum_{n=1}^N
(\tilde{x}_n-[\tilde{\mu}_\xvec]_n)^2[{\Cmat}_{\tilde{\xvec}\tilde{\xvec}}^{-1}]_{n,n}\hspace{1.25cm}\nonumber\\ + \frac{1}{2}\sum_{n=1}^N(\tilde{y}_n-
[\tilde{\gvec}(\Lmat,\tilde{x}_n\vvec_n)]_n
)^2[{\Cmat}_{\tilde\wvec\tilde\wvec}^{-1}]_{n,n}.
\eeqna
Thus, in the case described by Definition \ref{orthogonal_frequencies}, the objective function $Q_{freq}^{(d)}(\tilde{\xvec})$ is separable in the graph-frequency domain. The  component-wise formulation  of \eqref{Q_def_GSP_incremental2} results in a separable nonlinear WLS problem \cite{golub2003separable}, which
simplifies the task of designing the MAP estimator in the graph-frequency domain without the need for neglecting off-diagonal elements (as performed in the eGFD-MAP estimator).
 The minimization of \eqref{Q_def_GSP_incremental2} instead of  \eqref{Q_def_GSP} 
 results in $N$ independent optimization problems that 
 can result in an overall significant speedup of computations.
 Minimization of \eqref{Q_def_GSP_incremental2}  can be simplified even further if
 the graph input signal is a graph-bandlimited signal (see \eqref{bandlimited_def}).
In this case,  by substituting $\tilde{x}_n=0$, $\forall n>N_s$, the sums in \eqref{Q_def_GSP_incremental2} 
 can be computed only over $n=1,\ldots,N_s$.

The following theorem states sufficient conditions for the proposed 
estimators to coincide with the MAP estimator for the case of orthogonal-graph frequencies.
\begin{Theorem}\label{claim_coincides}
If the measurement function satisfies Definition \ref{orthogonal_frequencies} and, in addition, the following conditions hold:
\renewcommand{\theenumi}{C.\arabic{enumi}}
\begin{enumerate}
	\setcounter{enumi}{0}
    \item\label{cond1} The elements of the input graph signal, $\xvec$, are statistically independent in the graph-frequency domain, i.e. $\Cmat_{\tilde{\xvec}\tilde{\xvec}}$ is a diagonal matrix; 
    \item\label{cond2} The noise vector, $\wvec$, is uncorrelated in the graph-frequency domain, i.e. $\Cmat_{\tilde{\wvec}\tilde{\wvec}}$ is a diagonal matrix; 
\end{enumerate}
then the proposed eGFD-MAP, sGSP-MAP, and GSP-MAP estimators coincide with the MAP estimator.
\end{Theorem}
\begin{proof}
By substituting \eqref{separately_Model}, together with $\Cmat_{\tilde{\wvec}\tilde{\wvec}}^{-1}=\Dmat_{\tilde{\wvec}\tilde{\wvec}}^{(\text{inv})}$ and
$\Cmat_{\tilde{\xvec}\tilde{\xvec}}^{-1}=\Dmat_{\tilde{\xvec}\tilde{\xvec}}^{(\text{inv})}$ from Conditions \ref{cond1} and \ref{cond2}, respectively,  in the MAP objective function in the graph-frequency domain in \eqref{Q_def_GSP}, we obtain that in this case
\beqna
\label{Q_orth}
Q_{freq}(\tilde{\xvec})=
\frac{1}{2}(\tilde{\xvec}-\tilde{\muvec}_\xvec)^{T}{\Dmat}_{\tilde{\xvec}\tilde{\xvec}}^{(\text{inv})}(\tilde{\xvec}-\tilde{\muvec}_\xvec)\hspace{1.8cm}\nonumber\\ + \frac{1}{2}(\tilde{\yvec}-\tilde{\gvec}(\Lmat,\xvec))^{T}{\Dmat}_{\tilde{\wvec}\tilde{\wvec}}^{(\text{inv})}(\tilde\yvec-\tilde\gvec(\Lmat,\Vmat\tilde\xvec))
\nonumber\\=
\frac{1}{2}\sum_{n=1}^N
(\tilde{x}_n-[\tilde{\mu}_\xvec]_n)^2[\Cmat_{\tilde{\xvec}\tilde{\xvec}}^{-1}]_{n,n}\hspace{1.68cm}\nonumber\\ + \frac{1}{2}\sum_{n=1}^N(\tilde{y}_n-
[\tilde{\gvec}(\Lmat,\tilde{x}_n\vvec_n)]_n
)^2[\Cmat_{\tilde{\wvec}\tilde{\wvec}}^{-1}]_{n,n}.
\eeqna
Thus,  the eGFD-MAP objective function from \eqref{Q_def_GSP_incremental2} coincides with the MAP objective function from \eqref{Q_orth}.
From Definition \ref{orthogonal_frequencies}, we have that $\tilde{\Gmat}(\Lmat,\Vmat\tilde\xvec)=\bar\Dmat_{\tilde{\Gmat}}(\tilde\xvec)$. Therefore, we can conclude that the eGFD-MAP estimator coincides with the MAP estimator.
By substituting, $\Cmat_{\tilde{\xvec}\tilde{\xvec}}^{-1}=\Dmat_{\tilde{\xvec}\tilde{\xvec}}^{(\text{inv})}$,
$\Cmat_{\tilde{\wvec}\tilde{\wvec}}^{-1}={\Dmat}_{\tilde\wvec\tilde\wvec}^{(\text{inv})}$,
and $\tilde\Gmat(\Lmat,\Vmat\hat{\tilde{\xvec}}^{(t)})=\bar\Dmat_{\tilde{\Gmat}}(\hat{\tilde\xvec}^{(t)})$ in
the optimal graph filters from \eqref{filter1}
and \eqref{filter2}, and using the fact that $\Imat\circ\Amat={\text{ddiag}}(\Amat)$, ${\text{diag}}(\Dmat \avec)=\Dmat {\text{diag}}(\avec)$, and ${\text{ddiag}}(\Dmat)=\Dmat$ for any 
 matrix $\Amat$,
diagonal matrix, $\Dmat$, and vector $\avec$, we obtain the graph filters of the eGFD-MAP estimator from \eqref{f1_def} and \eqref{f2_def}. Thus, 
the GSP-MAP estimator coincides with the eGDF-MAP estimator, which is also the MAP estimator.
Moreover, since ${\Cmat}_{\tilde\xvec\tilde\xvec}^{-1}$ 
and $\big({\Cmat}_{\tilde\xvec\tilde\xvec}^{-1}+\tilde\Gmat^T(\Lmat,\Vmat\hat{\tilde{\xvec}}^{(t)}) 
{\Cmat}_{\tilde\wvec\tilde\wvec}^{-1}
\tilde\Gmat(\Lmat,\Vmat\hat{\tilde{\xvec}}^{(t)})\big)$ are diagonal matrices under the conditions of this theorem,
 one obtains 
 \be
 \label{ddiag_eq} {\text{diag}}\big(\Smat_{\hat{\tilde{\xvec}}\hat{\tilde{\xvec}}}^{(t)}{\Cmat}_{\tilde\xvec\tilde\xvec}^{-1}\big)= \text{diag}
\big({\text{ddiag}}\big(\Smat_{\hat{\tilde{\xvec}}\hat{\tilde{\xvec}}}^{(t)}\big){\Cmat}_{\tilde\xvec\tilde\xvec}^{-1}\big)
\ee
and 
\beqna
 \label{ddiag_eq2} 
\Smat_{\hat{\tilde{\xvec}}\hat{\tilde{\xvec}}}^{(t)}
\circ ({\Cmat}_{\tilde\xvec\tilde\xvec}^{-1}+ \tilde\Gmat^T(\Lmat,\Vmat\hat{\tilde{\xvec}}^{(t)}) 
{\Cmat}_{\tilde\wvec\tilde\wvec}^{-1}
\tilde\Gmat(\Lmat,\Vmat\hat{\tilde{\xvec}}^{(t)}))
\nonumber\\
={\text{ddiag}}\big(\Smat_{\hat{\tilde{\xvec}}\hat{\tilde{\xvec}}}^{(t)}\big)({\Cmat}_{\tilde\xvec\tilde\xvec}^{-1}+ \tilde\Gmat^T(\Lmat,\Vmat\hat{\tilde{\xvec}}^{(t)}) 
{\Cmat}_{\tilde\wvec\tilde\wvec}^{-1}
\tilde\Gmat(\Lmat,\Vmat\hat{\tilde{\xvec}}^{(t)})).
\eeqna
By substituting \eqref{ddiag_eq} and \eqref{ddiag_eq2} in the graph filter from \eqref{filter1_noE}, we obtain
\beqna
f_1^{\text{sGSP-MAP}}(\Lambdamat,\hat{\xvec}^{(t)} ) = 
 ({\Cmat}_{\tilde\xvec\tilde\xvec}^{-1} \hspace{3cm}\nonumber\\+ \tilde\Gmat^T(\Lmat,\Vmat\hat{\tilde{\xvec}}^{(t)}) 
{\Cmat}_{\tilde\wvec\tilde\wvec}^{-1}
\tilde\Gmat(\Lmat,\Vmat\hat{\tilde{\xvec}}^{(t)})\big)^{-1} {\Cmat}_{\tilde\xvec\tilde\xvec}^{-1},
\nonumber
\eeqna
where we use the fact that the elements of $\Smat_{\hat{\tilde{\xvec}}\hat{\tilde{\xvec}}}^{(t)}$ 
are nonzero with probability 1. Similarly, it can be shown that under the conditions of this theorem, the graph filter from \eqref{filter2_noE} satisfies 
\beqna
f_2^{\text{sGSP-MAP}}(\Lambdamat,\hat{\xvec}^{(t)} ) = 
 ({\Cmat}_{\tilde\xvec\tilde\xvec}^{-1} \hspace{3.6cm}\nonumber\\+ \tilde\Gmat^T(\Lmat,\Vmat\hat{\tilde{\xvec}}^{(t)}) 
{\Cmat}_{\tilde\wvec\tilde\wvec}^{-1}
\tilde\Gmat(\Lmat,\Vmat\hat{\tilde{\xvec}}^{(t)})\big)^{-1} {\Cmat}_{\tilde\wvec\tilde\wvec}^{-1}
\tilde\Gmat(\Lmat,\Vmat\hat{\tilde{\xvec}}^{(t)}).\nonumber
\eeqna
By substituting these filters in \eqref{opt_estimator_graph} we obtain the MAP estimator in the graph frequency domain. Hence, the sGSP-MAP estimator also coincides with the MAP estimator.
\end{proof}

The following corollary presents a special case of  Theorem \ref{claim_coincides}.
\begin{corollary}\label{claim_graphical_Model}
The proposed eGFD-MAP, sGSP-MAP, and GSP-MAP estimators coincide with the MAP estimator
if Conditions \ref{cond1} and \ref{cond2} hold, and 
	\renewcommand{\theenumi}{C.\arabic{enumi}}
\begin{enumerate}
	 \setcounter{enumi}{2}
    \item\label{cond3} The measurement function, $\gvec(\Lmat,\xvec)$, is the output of  a linear graph filter as defined in \eqref{laplacian_graph_filter}, i.e.
    \be
    \label{g_filter1}
    \gvec(\Lmat,\xvec) = \Vmat f(\Lambdamat)\Vmat^T \xvec.
    \ee
\end{enumerate}
\end{corollary}
\begin{proof}
By multiplying \eqref{g_filter1}  by $\Vmat^{T}$ and using the fact that $\Vmat$ is a unitary matrix, i.e. $\Vmat\Vmat^T = \Imat$, we obtain
\be
    \label{g_filter2}
    \tilde\gvec(\Lmat,\Vmat\tilde\xvec) = f(\Lambdamat) \tilde\xvec.
\ee
Since $f(\Lambdamat)$ is a diagonal matrix, we obtain that the measurement function satisfies  \eqref{separately_Model} in Definition \ref{orthogonal_frequencies}. Since, in addition, we assume that  Conditions \ref{cond1} and \ref{cond2}  of Theorem \ref{claim_coincides} hold,  this is a special case of Theorem \ref{claim_coincides}. Thus, the eGFD-MAP and GSP-MAP estimators coincide with the MAP estimator.
\end{proof}

The special case in Corollary \ref{claim_graphical_Model} fits the model behind
the graphical Wiener filter \cite{7891646,kroizer2021bayesian}, where
under Condition \ref{cond1}, the signal $ \xvec - \EX[\xvec]$ is a graph wide-sense stationary signal (see Definition 3 and Theorem 1 in \cite{7891646}).
Therefore, if the conditions of Corollary \ref{claim_graphical_Model} hold, an
 estimator that is obtained by minimizing \eqref{Q_def_GSP_incremental2} coincides with the graphical Wiener filter \cite{7891646}, the MAP estimator, the LMMSE estimator,   and the GSP-LMMSE estimator from \cite{kroizer2021bayesian}.
In this case,  the Gauss-Newton iterative approach converges in a single iteration with a step size $\alpha=1$, since $\gvec$ is a linear function (see Sec. 1.5.1 in \cite{BERTSEKAS_1999}).
\subsection{Computational complexity}\label{computational_complexity}
The computational complexity and the run-time of the proposed iterative estimators mainly depend on: 1) the total number of iterations until convergence of the estimator; 2) the total subiterations in the backtracking line search algorithm (Algorithm \ref{line_search_stages}); and 3) the matrix multiplications in the update step. If we assume that 1) and 2) are roughly similar among the different estimators (given that the parameters $\alpha^{(0)}$, $\gamma$, $\Delta$, and $\delta$ are the same among the estimators), then the differences in the complexity and run-time are due to the update steps of the different estimators. In addition, the computations of the inverse covariance matrices of $\xvec$ and $\wvec$ in the graph and graph-frequency domains, and the computation of the eigenvalue decomposition of the Laplacian matrix, which is of order $\mathcal{O}(N^3)$, can be done offline. Therefore, we do not consider them in the computational complexity. 

The update rule of the MAP estimator in \eqref{iteration_Newton2} consists of full matrix multiplications with a computational complexity of $\mathcal{O}(N^3)$, and inversion of an $N\times N$ full matrix, which has a complexity of $\mathcal{O}(N^3)$. This also holds for the implementation of the MAP estimator in the graph-frequency domain; however, as mentioned, one implementation of the MAP (in the graph or graph-frequency domain) may be more efficient than another. The computational complexity of the update steps of the sGSP-MAP and GSP-MAP estimators are similar to that of the MAP estimator, since the reconstruction of each filter in \eqref{filter1}, \eqref{filter2} and in \eqref{filter1_noE}, \eqref{filter2_noE} demands full matrix multiplications in the order of $\mathcal{O}(N^3)$, and inversion of a $N\times N$ matrix. The update step of the eGFD-MAP estimator in \eqref{27_old25} can be implemented in a vectorized form, and thus requires only $\mathcal{O}(N)$ multiplications, without the need to invert a matrix. In addition, in order to perform the update steps of the MAP and GSP-MAP estimators, it is necessary to calculate $N^2$ elements of the Jacobian matrix in \eqref{GGG} (or in \eqref{GGG_tilde}), compared to $N$ diagonal elements for the eGFD-MAP estimator in \eqref{27_old25}.
As the size of the network increases, the differences in the computational complexities of the estimators become more significant, so that in very large networks the MAP and GSP-MAP estimators may become intractable.
The total floating point operations (FLOPs) and order of the update steps of the different estimators is summarized in Table \ref{computational_complexity_tab}.
\begin{table}[hbt!]\centering
\begin{adjustbox}{max width=0.489\textwidth}
\renewcommand{\arraystretch}{2}
\begin{tabular}{|c|c|c|c|c|}
\hline
Estimator & MAP                & eGFD-MAP         & sGSP-MAP & GSP-MAP            \\ \hline
FLOPs     & $11N^{3}+2.5N^{2}+1.5N$ & $11N$ & $10N^3+10.5N^2+4N$ & $10N^3+6N^{2}+5N$ \\ \hline
Order     & $\mathcal{O}(N^3)$ & $\mathcal{O}(N)$ & $\mathcal{O}(N^3)$ & $\mathcal{O}(N^3)$\\ \hline
\end{tabular}
\end{adjustbox}
\caption{The FLOPs and order of complexity required for the update rules of the different estimators.}
\label{computational_complexity_tab}
 \vspace{-0.5cm}
\end{table}


\section{Simulation} \label{simulation}
In this section, we evaluate the performance of the different estimators for
a synthetic example in Subsection \ref{toy} and for the PSSE problem in Subsection \ref{seeting_sec}.
The estimators that are used in this section are:
\begin{itemize}
    \item MAP estimator, implemented by Algorithm \ref{algorithm_MAP},  while using \eqref{iteration_Newton2} (i.e. in the vertex domain); 
    \item  MAP estimator in the graph-frequency domain (MAP-FD), implemented by Algorithm \ref{algorithm_MAP};
    \item eGFD-MAP estimator, implemented by Algorithm
    \ref{algorithm_general};
        \item sGSP-MAP estimator, implemented by Algorithm \ref{new_alg} with Option A;
    \item GSP-MAP estimator, implemented by Algorithm \ref{new_alg} with Option B;
    \item LMMSE estimator and GSP-LMMSE estimator from \cite{kroizer2021bayesian}; In  Subsection \ref{toy}, the analytic-versions of LMMSE and GSP-LMMSE estimators are implemented, while
    in Subsections 
    \ref{seeting_sec} and \ref{sensativity2initialization} these estimators are implemented by their sample-mean version with $P=500$ training samples (i.e. where the intractable covariance matrices are replaced by the sample covariance matrices), as discussed in \cite{kroizer2021bayesian}. 
\end{itemize}
All iterative estimators (MAP, eGFD-MAP, sGSP-MAP, and GSP-MAP) were initialized by the GSP-LMMSE estimator,
unless written otherwise.
The performance of the different estimators is calculated by performing 10,000 Monte Carlo simulations for each scenario.

\subsection{Example A: Synthetic data - orthogonal graph frequencies}
\label{toy}
In this example we use random graphs that were generated by the Watts-Strogatz small-world graph
model \cite{Watts_Strogatz}, with different number of vertices, $N$, and a mean degree of $K=5$.
We evaluate the different estimators under the model from \eqref{Model} with the following nonlinear measurement function in the graph-frequency domain:
\begin{equation}\label{sep_func}
    [\tilde\gvec(\Lmat,\Vmat\tilde\xvec)]_{n} = \tilde\xvec_{n}^{3} \quad n=1,\ldots,N.
\end{equation}
We assume that the {\em{a-priori}} probability density function (pdf) of $\xvec$ is given by $\tilde{\xvec} \sim \pazocal{N}(\zerovec,\sigma_{\xvec}^{2}\Imat)$. The noise from \eqref{Model}, $\wvec$, is white Gaussian noise with $\Cmat_{\wvec\wvec}=\sigma_{\wvec}^{2}\Imat$. In the simulations, we use $\sigma_{\xvec}^{2}=0.5$ and $\sigma_{\wvec}^{2}=0.05$.
It can be verified that this nonlinear model has {\em{orthogonal graph frequencies}}, as defined in Definition \ref{orthogonal_frequencies}, and that the conditions of Theorem \ref{claim_coincides} are satisfied since $\tilde\xvec$ and $\tilde\wvec$ are white Gaussian noise signals.
Thus, according to Theorem \ref{claim_coincides} the MAP estimator coincides with the eGFD-MAP and the GSP-MAP estimators. 
Under this setting, the LMMSE and GSP-LMMSE estimators coincide and can be computed analytically by using  $\Cmat_{\yvec\yvec}=(\sigma_{\wvec}^{2}+15\sigma_{\xvec}^{6})\Imat$ and $\Cmat_{\xvec\yvec}=3\sigma_{\xvec}^{4}\Imat$. 
In this subsection, we also present the MAP-FD estimator in order to have a fair run-time comparison.

In Fig. \ref{fig:MSE_SYN} we present the normalized MSE (NMSE), i.e. the MSE divided by the number of vertices, $N$, of the different estimators versus  $N$.
In this case,
 the nonlinear estimators  (MAP, MAP-FD, eGFD-MAP, sGSP-MAP, and GSP-MAP estimators) significantly outperform the linear estimators (LMMSE and GSP-LMMSE), which have an almost constant NMSE for any $N$ around the value $0.208$ with a standard deviation of $0.001$.  Therefore, and due to resolution reasons, the linear estimators are omitted from this figure. 
Since in this case  the conditions  of Theorem \ref{claim_coincides} are satisfied, the MSEs of the nonlinear estimators (MAP, MAP-FD, eGFD-MAP, sGSP-MAP,
 and GSP-MAP estimators) are all equal, as expected from Theorem \ref{claim_coincides}.
 It can be seen that the MSE of the nonlinear estimators increases as  $N$ increases. 
 \vspace{-0.25cm}
 \begin{figure}[hbt]
  \centering
\includegraphics[width=0.65\linewidth]{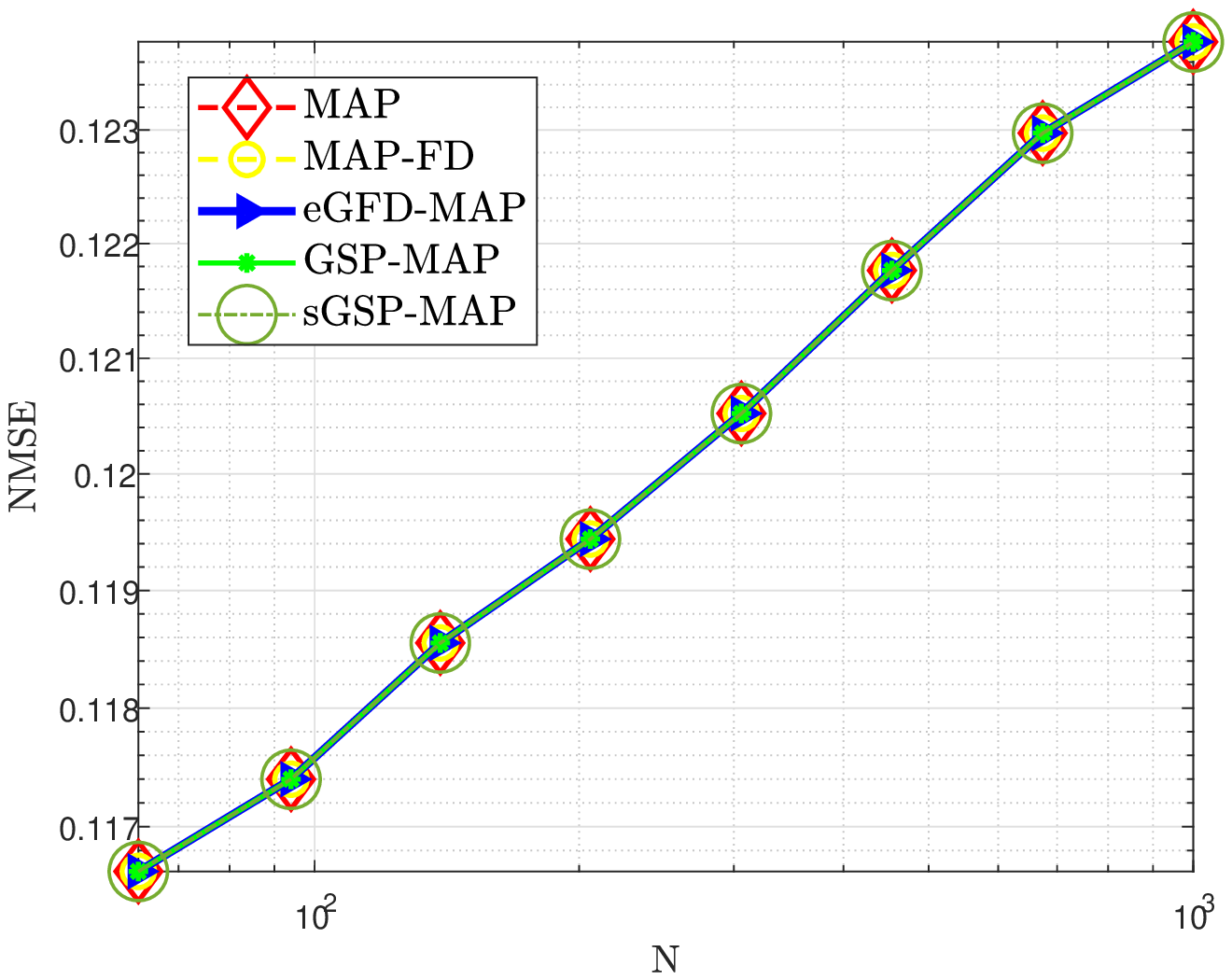}
  \caption{Example A: The NMSE of the different estimators versus  $N$. }
  \label{fig:MSE_SYN}
  \end{figure}

In Fig. \ref{fig:TIME_SYN}  we present the averaged run-time of any estimator till convergence versus $N$, which is evaluated using MATLAB on an Intel(R) Core(TM) i7-6800K CPU computer, 3.4 GHz. 
 It can be seen that the eGFD-MAP estimator, which does not require matrix inversion per iteration, has the lowest run-time. 
This complexity reduction becomes more evident as the dimension of
the network increases.
Hence, the eGFD-MAP estimator is a good alternative to the MAP estimator in cases where the measurement function is close to separable in the graph-frequency domain, i.e. has almost {\em{orthogonal graph frequencies}}. It can be seen that, in this case, implementation in the graph frequency domain has lower computational complexity, since the run-time of the MAP-FD estimator is lower than that of the MAP estimator. In addition, the run-time of the sGSP-MAP estimator is higher than the run-time of the GSP-MAP estimator and similar to that of the MAP estimator, since it involves the additional computation of two covariance matrices and the factors $\Smat_{\hat{\tilde{\xvec}}\hat{\tilde{\xvec}}}^{(t)}{\Cmat}_{\tilde\xvec\tilde\xvec}^{-1}$ and $\Smat_{\tilde{\wvec}\tilde{\wvec}}^{(t)}{\Cmat}_{\tilde\wvec\tilde\wvec}^{-1}
\tilde\Gmat(\Lmat,\Vmat\hat{\tilde{\xvec}}^{(t)})$  at each iteration.
The averaged run-time of the linear estimators (not shown in this figure due to resolution reasons) is significantly lower than the iterative estimators (around $1 \cdot 10^{-6}$ [sec]). This run-time would significantly increase in complex models where there is a need to compute the sample covariance matrices, as in the following subsection. 
\begin{figure}[hbt]
  \centering
\includegraphics[width=0.65\linewidth]{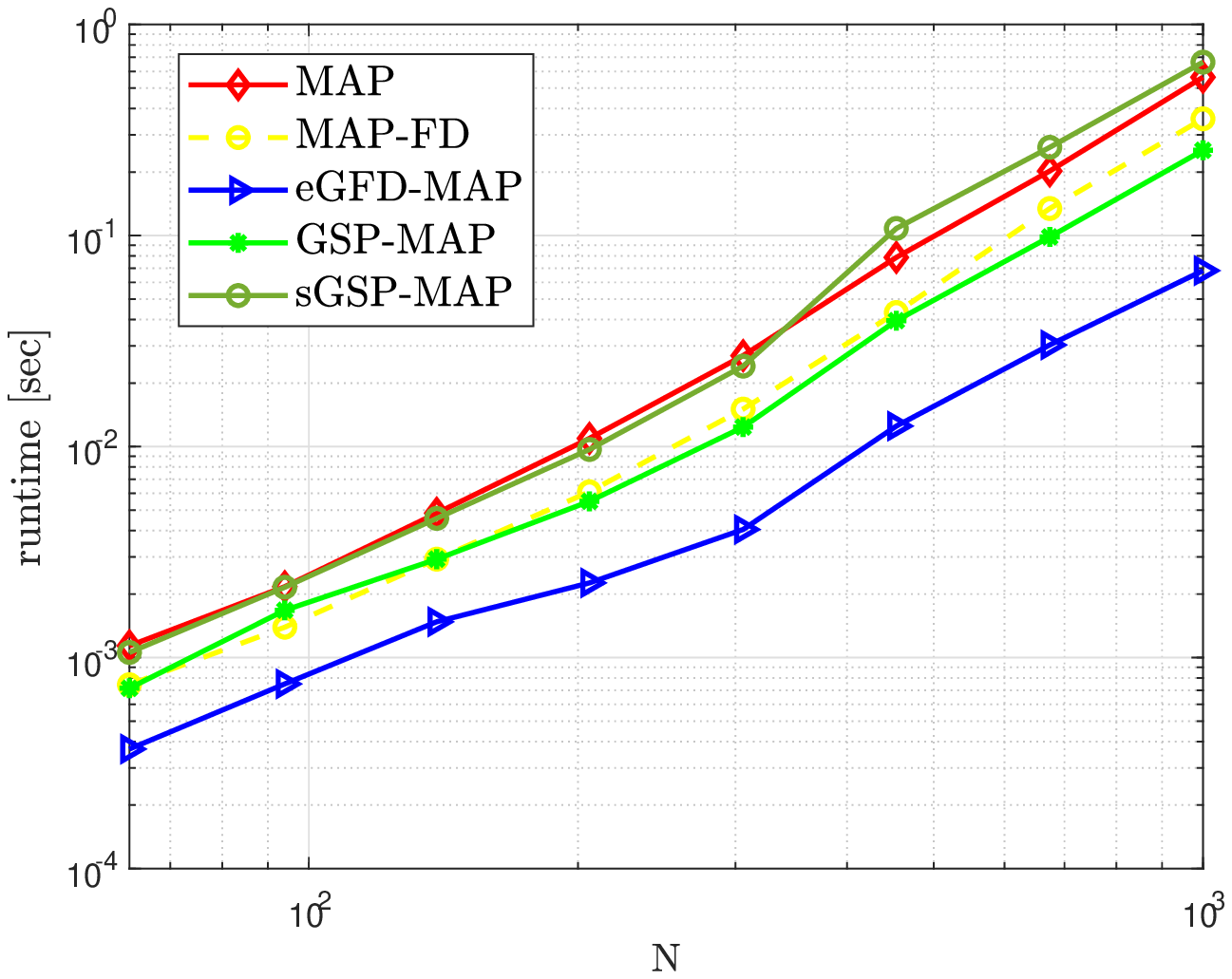}
  \caption{Example A: The time complexity of the different estimators versus  $N$.}
  \label{fig:TIME_SYN}
  \end{figure}
\subsection{Example B: PSSE in electrical networks}
\label{seeting_sec}
A power system can be represented as an undirected weighted graph, ${\pazocal{G}}({\pazocal{V}},\xi,\Wmat)$, where the set of vertices, $\pazocal{V}$, is the set of buses (generators/loads) and the edge set, $\xi$, is the set of transmission lines between these buses.
The measurement vector of the active powers at the buses, $\yvec$,
can be described by the model in \eqref{Model} with the  measurement function \cite{Abur_book}:
\begin{eqnarray} \label{g_AC}
\left[\gvec(\Lmat,\xvec)\right]_n\
 \define \sum\nolimits_{m=1}^N |v_n||v_m|(G_{n,m}\cos(x_n -x_{m})
\nonumber\\+B_{n,m}\sin(x_n -x_{m})),\hspace{1.5cm}
\end{eqnarray}
$n=1,\ldots,N$. Here, $x_n$ and $|v_n|$ are the voltage phase and amplitude at the $n$th bus,
and $G_{n,m}$ and $B_{n,m}$ are the conductance and susceptance of the transmission line between the buses $n$ and $m$ \cite{Abur_book}, where $(n,m)\in\xi$.
We assume that $|v_n|=1$, which is a common assumption in normalized power systems \cite{Abur_book}, and that $G_{n,m}$ and $B_{n,m}$ are all known. In the graph modeling of the electrical network, the Laplacian matrix, $\Lmat$, is usually constructed as follows (Subsection II-C in \cite{Grotas2019}):
\beqna
\label{L_B}
[\Lmat]_{n,k}=
\left\{
     \begin{array}{lr}
       -B_{n,k} , & k \neq n\\
    \sum_{m=1,~m\neq n}^{N} B_{n,m} , & k = n.  
     \end{array}\right.
\end{eqnarray}

The goal of PSSE is to recover the state vector, $\xvec$, from the power measurements,
$\yvec$, described by \eqref{Model} with the nonlinear measurement function 
$\gvec(\Lmat,\xvec)$ in \eqref{g_AC}.
This estimation task is known to be an NP-hard problem \cite{bienstock2019strong},
and is essential for various monitoring purposes 
\cite{Abur_book}.
PSSE is traditionally solved by iterative methods, 
where the Gauss-Newton method is a commonly-used choice for this task \cite{Abur_book,Monticelli_2000,Cosovic_Vukobratovic_2019}.
By substituting \eqref{g_AC} in \eqref{GGG}, we obtain that the associated Jacobian matrix as
\beqna
\label{Jacobian_g}
[\Gmat(\Lmat,\xvec)]_{n,k}= 
\frac{\partial \gvec_{n}(\Lmat,\xvec)}{\partial x_{k}}=
\hspace{4cm}
\nonumber\\\hspace{-0.65cm}\left\{
     \begin{array}{lr}
       G_{n,k} \sin(x_{n}-x_{k})-B_{n,k} \cos(x_{n}-x_{k}), & k \neq n\\
       \mathlarger{\sum}_{\underset{m\neq n}{m=1}}^{N} -G_{n,m} \sin(x_{n}-x_{m})+B_{n,m} \cos(x_{n}-x_{m}), & k = n  
     \end{array}\right.\hspace{-0.15cm}.\hspace{-0.1cm}
\end{eqnarray}
In particular, since $G_{n,k}=0$  and $B_{n,k}=0$ for any $(n,k) \notin \xi$, it can be seen that $[\Gmat(\Lmat,\xvec)]_{n,k}= 0$ for any $(n,k) \notin \xi$. 
It should be noted that, in the general case, \eqref{Jacobian_g} implies that the conditions of Theorem \ref{claim_coincides} are not satisfied.

The input graph signal, $\xvec$, has been shown to be smooth w.r.t. the graph \cite{drayer2018detection,dabush2021state}.
Therefore, we model the distribution of  $\xvec$ in the graph-frequency domain as a smooth, normal distribution \cite{Dong_Vandergheynst_2016,ramezani2019graph}, as defined in \eqref{x_distribution} in Appendix \ref{WLS_Interpretation}.
Since in this case study $\xvec$ represents phases, the value of $\beta$ in \eqref{x_distribution} is taken such that the probability that the elements of $|\xvec|$ are larger than $\pi$ is smaller than 0.01.
We assume that the covariance matrix of the noise,  $\wvec$, from the model in \eqref{Model} is 
$\Cmat_{\wvec\wvec}=\sigma_{\wvec}^2\Imat$.
The values of the different physical parameters in \eqref{g_AC}, e.g. the conductance and the susceptance matrices, and the voltages,  are all taken from the IEEE   118-bus test case 
\cite{iEEEdata}\footnote{We repeated the simulations in this example for the 57-bus test case \cite{iEEEdata}, where $N=57$, and obtained similar results. In this test case, not shown here, the GSP-MAP estimator outperforms the eGFD-MAP estimator.}, 
 which is a simple approximation of the American Electric Power system (in the U.S. Midwest) as of December 1962. This system has $N=118$ buses (vertices) with power measurements, contains generators, synchronous condensers, lines, transformers, and loads.
In Table \ref{hyperparameters} we present the hyperparameters used in this and in the following subsection.

\begin{table}[hbt!]\centering
\renewcommand{\arraystretch}{0.9}
\begin{tabular}{|c|c|c|c|}
\hline
Estimator & MAP & eGFD-MAP & GSP-MAP \\ \hline
$\alpha_{0}$         & 0.5     & 0.5          & 0.5     \\ \hline
$\gamma$         & 0.83     & 0.83          & 0.83     \\ \hline
$\delta$         & 0.1     & 0.1          & 0.1     \\ \hline
$\Delta$         & 0.01     & 0.01         & 0.01    \\ \hline
\end{tabular}
\caption{The hyperparameters used for the simulations in Subsections \ref{seeting_sec} and \ref{sensativity2initialization}.
}
\label{hyperparameters}
\end{table}
Finally, it can be seen that in the model in \eqref{g_AC}, there is an inherent ambiguity, and one can recover the phases in $\xvec$ only up to modulo $2\pi$ errors. Therefore,  in the following simulations, the error is presented in terms of normalized mean-squared-periodic-error (NMSPE)  \cite{Routtenberg_Tabrikian_Bayesian}:
\begin{equation}
\label{eq:rmspe}
    {\text{NMSPE}}(\xvec, \hat{\xvec}) = \frac{1}{N}{\rm{E}}[ \left({\text{mod}}_{2\pi}(\xvec-\hat{\xvec})\right)^2],
\end{equation}
where  $\text{mod}_{\pi}$ denotes the element-wise modulo operator. 
The units of the NMSPE are $[\text{rad}^{2}]$. 
Similarly, the periodic bias can be computed \cite{Routtenberg_Tabrikian_Bayesian}. However, due to space limitations, and since the bias is negligible for the simulations below, it is omitted from the results below. 
In this case Conditions 1) and 2) of Claim \ref{Q_lin_Q_freq_coincide} hold, and thus, the objective functions from \eqref{Q_def_GSP} and \eqref{Q_def_GSP_lin_incremental} coincide, which gives an advantage to the eGFD-MAP estimator over the GSP-MAP estimator.

In Fig. \ref{fig:Diff N} we present the NMSPE of the different estimators
versus the inverse of the noise variance, $\frac{1}{\sigma_{\wvec}^2}$, for the 118-bus case. 
It can be seen that the MAP, eGFD-MAP, sGSP-MAP, and  GSP-MAP estimators significantly outperform the linear estimators (LMMSE and GSP-LMMSE estimators) for $\sigma_w^2<1$.
As the noise variance, $\sigma_{\tilde\wvec}^2$, increases, the NMSPE of the LMMSE and the GSP-LMMSE achieved the NMSPE of the nonlinear estimators, since in the presence of significant measurement noise, 
all estimators are reduced to the prior-mean estimator, $\muvec_\xvec$, which is a linear estimator.
Therefore, this figure shows that we can achieve good performance with the low-complexity eGFD-MAP and with the GSP-MAP estimators even when the measurement function does not have orthogonal graph frequencies. This is noteworthy, considering the restriction of the output of graph filters of the GSP estimators.  The sGSP-MAP archives the performance of the MAP estimator since it shares the same objective function.
Furthermore, it can be seen that the eGFD-MAP estimator outperforms the GSP-MAP estimator in this figure. 
However, it should be noted that in other tested cases (not shown here), the GSP-MAP estimator outperforms the eGFD-MAP estimator.  The sGSP-MAP estimator achieves similar NMSPE to those of the MAP estimator, while being an output of graph filters. 
\begin{figure}[hbt]
\centering
{ \includegraphics[width=0.75\linewidth]{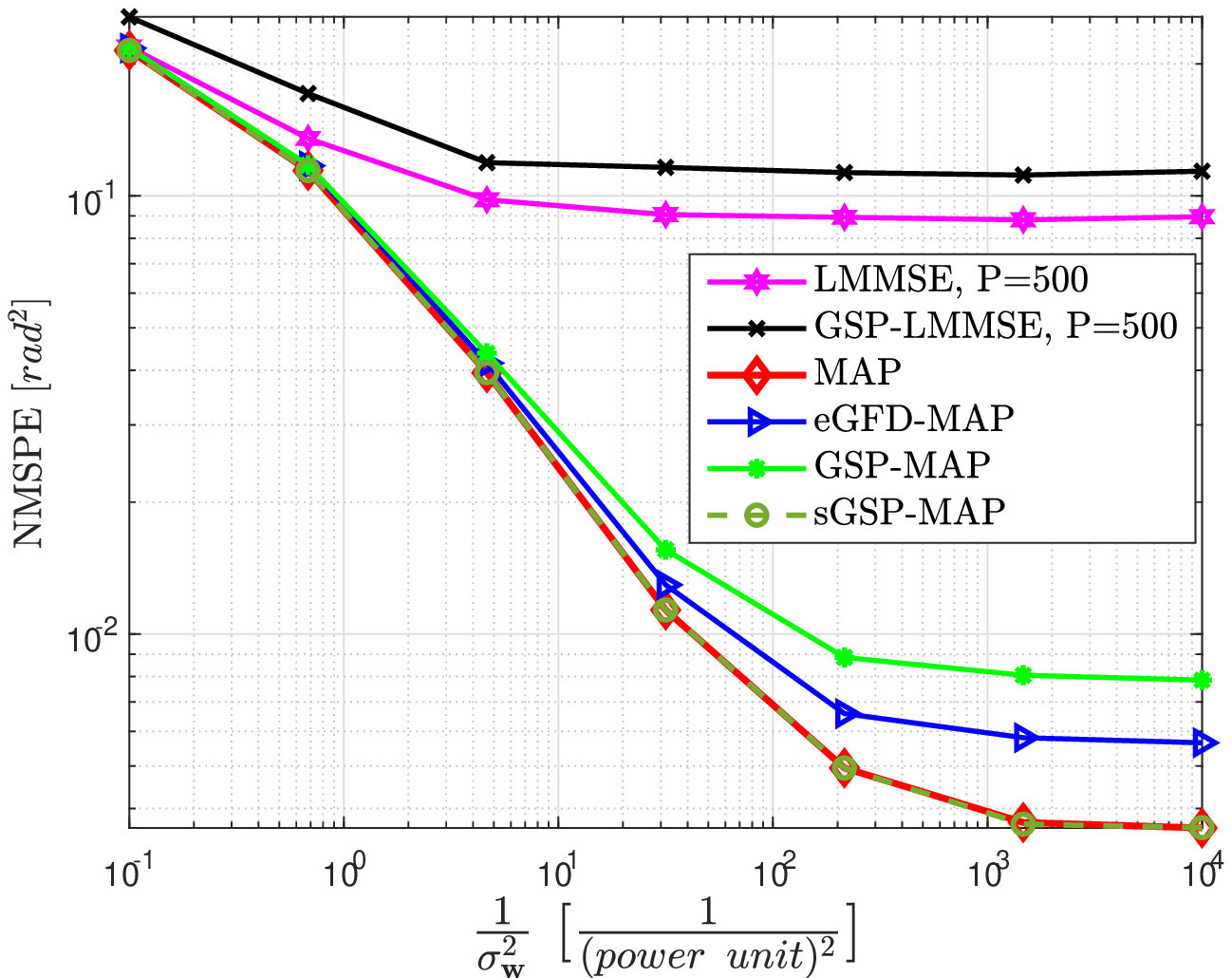}}
\caption{The NMSPE of the different estimators versus $\frac{1}{\sigma_{\wvec}^2}$ for 118-bus test case with $\beta = 3$.\label{fig:Diff N}}
\end{figure}

In order to examine the robustness of the iterative estimators under different phase distributions, and in particular under different levels of separability of the measurement function, we changed the value of $\beta$ from \eqref{x_distribution}, which directly affects the separability of the measurement function in the sense of Definition \ref{orthogonal_frequencies}.
As $\beta$ increases, the total variation of the signal $\xvec$, $\xvec^{T}\Lmat\xvec$, increases. 
From the distribution in \eqref{x_distribution},  and recalling that $\xvec=\Vmat\tilde\xvec$, the variance of $x_{i}$ can be expressed as
\begin{equation}\label{phase_var}
    \text{var}\big({x_{i}}\big)=\beta\ \sum\nolimits_{j=1}^{N} \frac{V^{2}_{i,j}}{\lambda_{j}}.
\end{equation}
Thus,  as $\beta$ decreases, the phases  are more concentrated around their mean, which is assumed to be the same for any $n=1,\ldots,N$.
Hence, as $\beta$ decreases, the difference $x_n-x_m$ is  smaller, and when it is sufficiently small, it is possible to use the first-order Taylor approximation: 
\begin{equation}
\label{approx}
    \sin(x_n-x_m)\approx x_n-x_m,~
    \cos(x_n-x_m)\approx 1.
\end{equation}
Therefore, by substituting \eqref{approx} in \eqref{g_AC} and taking the derivative w.r.t. $\xvec$, we obtain that the  Jacobian of $\gvec(\Lmat,\xvec)$ from \eqref{g_AC} can be approximated by
\begin{equation}
\label{G_L}
    \Gmat(\Lmat,\xvec) \approx  -\Lmat = \Vmat\Lambdamat\Vmat^{T},
\end{equation}
where we used \eqref{L_B} and the eigenvalue decomposition of $\Lmat$.
By multiplying \eqref{G_L} by $\Vmat^T$ and $\Vmat$ from the left and right, respectively, and by using \eqref{GGG_tilde}, one obtains
\begin{equation}\label{G_tilde_small_beta}
    \tilde\Gmat(\Lmat,\Vmat\tilde\xvec)  \approx -\Lambdamat, 
\end{equation}
i.e. $\tilde\Gmat(\Lmat,\Vmat\tilde\xvec)$ is a diagonal matrix.
Equation \eqref{G_tilde_small_beta} implies that as $\beta$ decreases, $\tilde\gvec(\Lmat,\Vmat\tilde\xvec)$ becomes more separable in the sense of Definition \ref{orthogonal_frequencies}. 
However, for a general $\beta$,  $\tilde\gvec(\Lmat,\Vmat\tilde\xvec)$ may not be separable as required by Definition \ref{orthogonal_frequencies}.

In Fig. \ref{fig:beta_118} we present the NMSPE versus the parameter $\beta$. It can be seen that
the NMSPE of the different estimators increases as $\beta$ increases, since $\beta$ is proportional to the variance of the unknown parameters. Moreover, the iterative estimators have significantly lower NMSPE than the linear estimators for most $\beta$ values. 
In addition, it can be seen that the MAP, eGFD-MAP, and sGSP-MAP estimators have NMSPE values that are similar to that of the GSP-MAP estimator until a certain value of $\beta$. Above this value, the GSP-MAP estimator outperforms the other estimators, as the eGFD-MAP estimator performance is degraded due to the fact that for large values of $\beta$, $\tilde\Gmat(\Lmat,\Vmat\tilde\xvec)$ cannot be approximated with small errors.
Finally, it should be noted that, in this case, the proposed estimators perform well even when the conditions of Theorem \ref{claim_coincides} are not satisfied.
\begin{figure}[hbt!]
\centering
{ \includegraphics[width=0.75\linewidth]{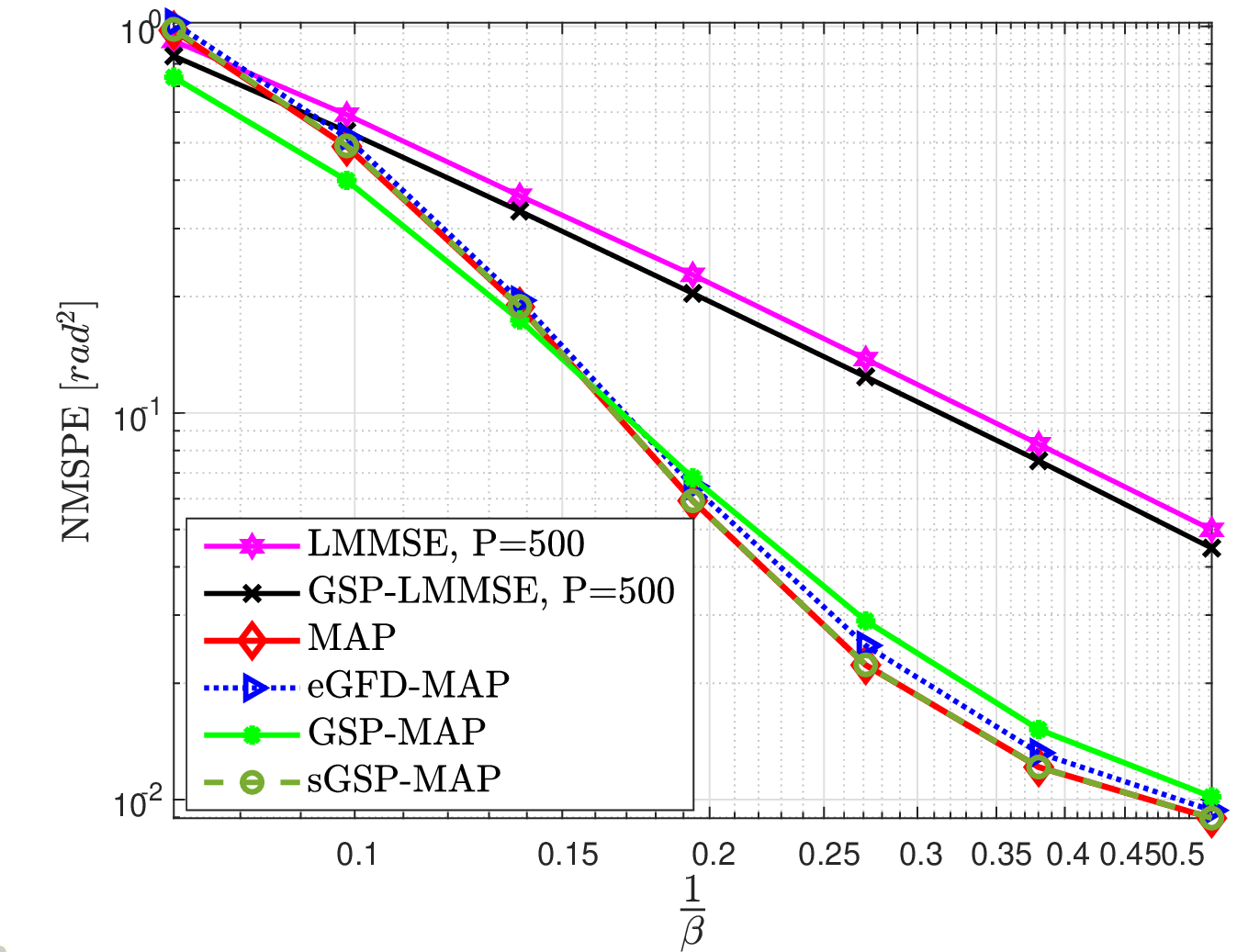}}
\caption{The NMSPE of the different estimators versus $\frac{1}{\beta}$  for 118-bus test case, where  ${\sigma}_{\wvec}^2=0.05$.\label{fig:beta_118}}
\end{figure}

\subsection{Sensitivity to initialization}\label{sensativity2initialization}
The implementation of the MAP estimator by the Gauss-Newton method is known to be sensitive to the initialization of the algorithm \cite{Fatemi_Svensson_Morelande2012}. In this subsection, we examine the robustness of the MAP, eGFD-MAP, sGSP-MAP, and GSP-MAP  estimators to perturbed initialization under the setting of the PSSE problem from Subsection \ref{seeting_sec}.
\subsubsection{Scenario I -  noisy initialization}
In the first scenario we use perturbed initialization, in which the estimators are initialized with
\begin{equation}
    \hat{\xvec}^{(0)} = \Vmat(\tilde\xvec_{0} + \tilde\pvec_{0}), ~~~\tilde\pvec_{0} \sim \pazocal{N}(\zerovec,\sigma_{p}^{2}\Imat),
\end{equation}
where $\tilde\xvec_{0}$ is the original initialization (i.e. the prior mean in the graph-frequency domain, $\tilde\muvec_{\xvec}$, or the GSP-LMMSE estimator in the graph frequency domain), and $\tilde\pvec_{0}$ is zero-mean Gaussian noise with covariance $\sigma_{p}^{2}\Imat$. Thus, as $\sigma_p$ increases, the initialization becomes more perturbed. For $\sigma_p\rightarrow 0$, i.e. no perturbation, we obtain the results from Subsection \ref{seeting_sec}.
The NMSPE of the estimators versus $\frac{1}{\sigma_{p}^{2}}$ is  presented   in Fig. 
\ref{fig:noise_INIT_118} for the 
118-bus test case with this perturbed initialization. 
It can be seen that the NMSPE decreases as $\sigma_p$ decreases, due to the influence of bad initialization. In addition, since the NMSPE from \eqref{eq:rmspe} is bounded, for large values of $\sigma_p$ all the estimators converge to the value of random estimation, $\frac{\pi^{2}}{{3}}$.
\begin{figure}[hbt!]
\centering
{ \includegraphics[width=0.75\linewidth]{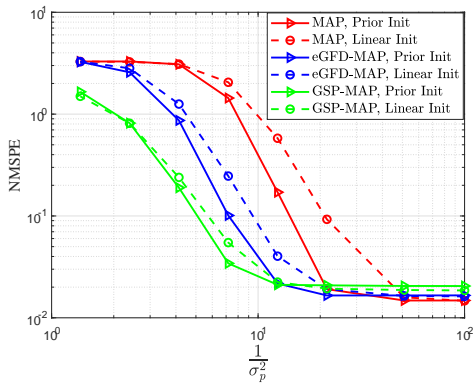}}
\caption{The NMSPE of the different estimators versus $\frac{1}{\sigma_{p}^{2}}$, where the noisy GSP-LMMSE and the noisy prior mean were used to initialize the estimators, where ${\sigma}_{\wvec}^2=0.05$ and $\beta=3$.}\label{fig:noise_INIT_118}
\end{figure}

\subsubsection{Scenario II -  Perturbed topology}
In the second scenario, the GSP-LMMSE and the LMMSE estimators were calculated under a change in the topology. In particular,  the sample-mean versions of these estimators were calculated under a misspecified model. While the true dataset was generated with a given topology, the linear estimators assume a different topology obtained from the original topology after removing $M$ edges. Then,  the misspecified sample-mean GSP-LMMSE estimator was used as the initial estimator for the iterative estimators (similar results were obtained by using the misspecified sample-mean LMMSE estimator). Thus, this scenario describes a perturbed initialization, which affects the initialization. 
The results are presented in  Fig. \ref{fig:MSE_PER_LIN_INIT_118} versus the number of removed edges, $M$. For $M= 0$, i.e. no misspecification in the topology that was used to initialize the algorithms, we obtain the results from Subsection \ref{seeting_sec}.
\begin{figure}[hbt!]
\centering
{ \includegraphics[width=0.75\linewidth]{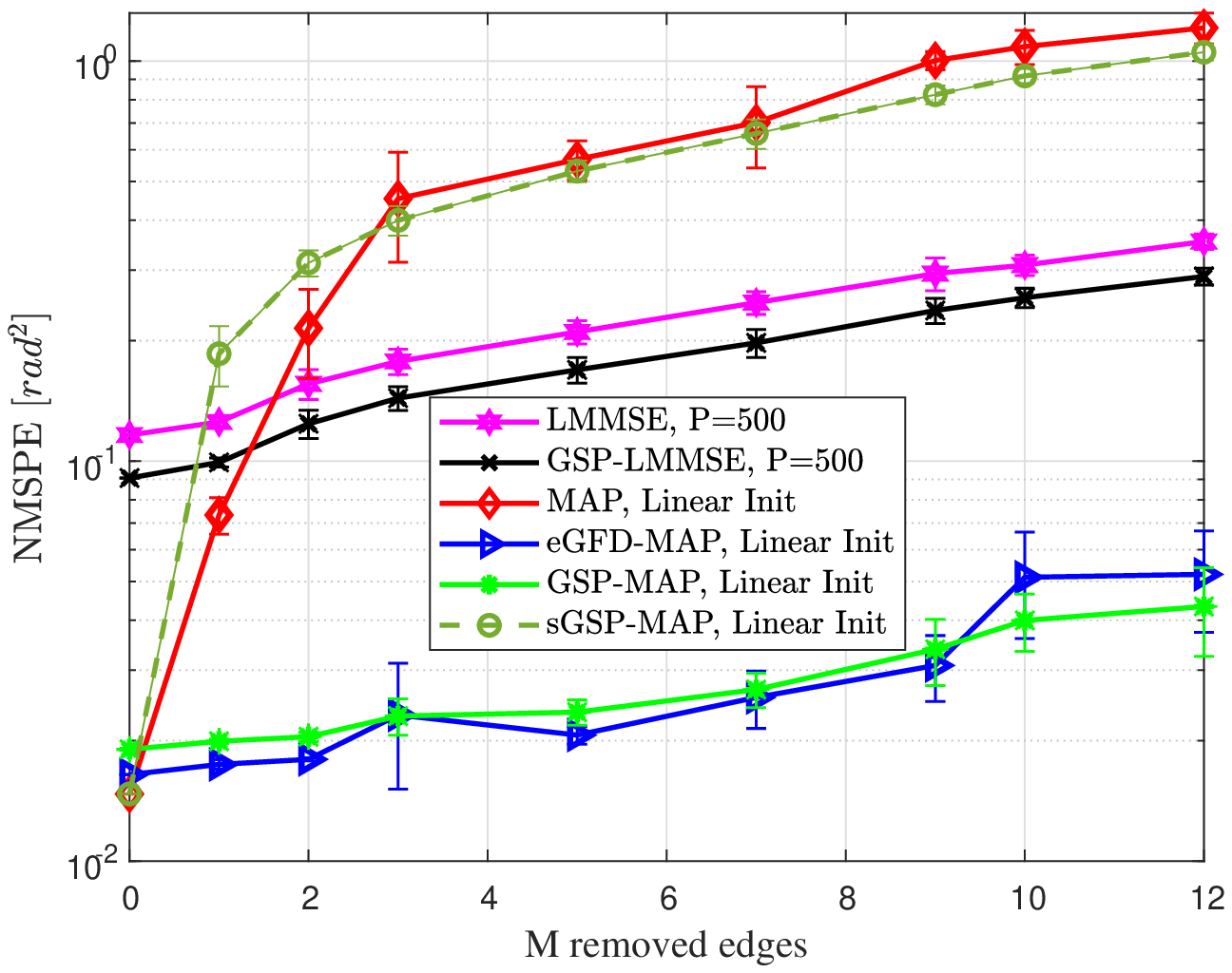}}
\caption{The NMSPE of the different estimators versus the number of removed edges $M$, where the linear estimators (that were used to initialize the iterative estimators) were perturbed by removing edges,  and ${\sigma}_{\wvec}^2=0.05$, $\beta=3$.}\label{fig:MSE_PER_LIN_INIT_118}
\end{figure}

\vspace{-0.2cm}
It can be seen in Figs. \ref{fig:noise_INIT_118} and \ref{fig:MSE_PER_LIN_INIT_118}, associated with Scenarios I and II, respectively, that the eGFD-MAP and GSP-MAP estimators are more robust to perturbed initialization than the sGSP-MAP and MAP estimators. Specifically, when the noise variance is small or the number of removed edges is low, the NMSPE of the MAP and the sGSP-MAP estimators are significantly larger than the NMSPE of the proposed eGFD-MAP and GSP-MAP estimators.
Thus, we can conclude that the methods that obtain superior robustness to graph perturbation leverage a combination of integrated graph structural information and utilization of the expectation of the MAP objective function. Moreover, the robustness observed in the low-complexity eGFD-MAP estimator can be attributed to the necessity of estimating fewer parameters and its function as a preconditioning approach, as elaborated upon in the concluding remarks of  Subsection \ref{eGFD-MAP_EST}.
When comparing the initialization methods, it can be seen that the eGFD-MAP and GSP-MAP estimators are more robust to the initialization for both initialization methods (GSP-LMMSE and prior mean). Thus, we can conclude that when the initialization is problematic, the proposed methods that use GSP information and preconditioning are preferable.  Furthermore, by comparing the robustness of the sGSP-MAP and GSP-MAP estimators in Figs. \ref{fig:noise_INIT_118} and \ref{fig:MSE_PER_LIN_INIT_118}, 
it can be seen that 
the restriction of the MAP iteration to be the output of two graph filters does not degrade the estimator performance, while the minimization of the expected objective function of the  GSP-MAP estimator enhances the resilience.     

\section{Conclusion}
\label{conclusion}
In this paper, we discuss the recovery of random graph signals from nonlinear measurements using a GSP-based MAP approach. We formulate the MAP estimator via the Gauss-Newton method implementation in both the vertex and the graph-frequency domains, 
which leads to the same estimator, but the efficiency and convergence rate of each implementation may be different. 
In order to accommodate the complexity and sensitivity to initialization of the Gauss-Newton MAP estimator, we develop the eGFD-MAP estimator that does not require a matrix inversion per iteration and updates the graph signal elements in the graph-frequency domain
independently. Thus, it can be interpreted as a preconditioning approach, which improves the convergence and robustness of the algorithm.
In addition, we derive the iterative sGSP-MAP and the GSP-MAP estimators with update equations that are composed of the output of two graph filters that are optimal in the sense of the non-expected and expected objective functions, respectively.
We show that when the measurement function is separable in the graph-frequency domain, and the input graph signal and the noise are uncorrelated, the proposed eGFD-MAP, sGSP-MAP, and  GSP-MAP estimators coincide with the MAP estimator. 

Our numerical simulations show that the proposed estimators outperform the linear estimators.
In addition, the eGFD-MAP and the GSP-MAP estimators achieve similar MSEs,  while the eGFD-MAP estimator has the lowest computational complexity, rendering its calculation tractable in large networks. The sGSP-MAP estimator almost achieves the performance of the MAP estimator, while keeping the GSP structure. However, it is sensitive and less robust to bad initialization.
Finally, we examine the sensitivity to initialization of the Gauss-Newton MAP estimator and show that the eGFD-MAP and the GSP-MAP estimators are significantly more robust to perturbed and noisy initialization.
Future work includes 
extension to dynamic systems \cite{10096261} and 
the implementation of the new GSP estimators with graph filter parameterization in order to increase the robustness to topology changes and to enable distributed implementations of the proposed estimators.


\appendices
	\renewcommand{\thesectiondis}[2]{\Alph{section}:}
	
\section{Interpretation of the MAP Estimation Problem as a Regularized WLS Problem}\label{WLS_Interpretation}
We consider the special case where the distribution of the input graph signal, $\xvec$, in the graph-frequency domain is a smooth, zero-mean Gaussian distribution \cite{Dong_Vandergheynst_2016,ramezani2019graph}:
\begin{equation} \label{x_distribution}
		\tilde{\xvec} \sim \pazocal{N}(\zerovec,\beta \Lambdamat^{\dagger}),
\end{equation}
where $\beta$ is the smoothness level.
In this case, $\muvec_\xvec=\zerovec$ and ${\Cmat}_{\xvec\xvec}=\beta \Lmat^\dagger$, where $(\cdot)^\dagger$ denotes the  pseudo-inverse operator. By substituting this prior in \eqref{Q_def}, we obtain
\beqna
\label{Q_def_smooth_prior}
Q(\xvec)=
\frac{1}{2}\beta \xvec^{T}\Lmat\xvec + \frac{1}{2}(\yvec-\gvec(\Lmat,\xvec))^{T}{\Cmat}_{\wvec\wvec}^{-1}(\yvec-\gvec(\Lmat,\xvec)).
\eeqna
 The left  term  on the r.h.s. of \eqref{Q_def_smooth_prior}, $\frac{1}{2}\beta \xvec^{T}\Lmat\xvec$, can be interpreted as a
regularization term,   which corresponds to assuming that the graph signal,
$\xvec$, is smooth on the graph.

If we take the model in \eqref{Model} where the prior information can be neglected (e.g. when ${\Cmat}_{\xvec\xvec}$ is significantly larger than ${\Cmat}_{\wvec\wvec}$ in the positive-definite matrix sense), then we can develop the WLS estimator instead of the MAP estimator, as considered, for example, in \cite{Sahu_Kar_Moura_Poor2016}. In this case, we do not have the first term in the objective functions (e.g. 
$\frac{1}{2}(\xvec-\muvec_\xvec)^{T}{\Cmat}_{\xvec\xvec}^{-1}(\xvec-\muvec_\xvec)$ in \eqref{Q_def}). However, we can add a Laplacian regularization term, $\mu \xvec^{T}\Lmat\xvec$, to obtain the regularized WLS problem in \eqref{Q_def_smooth_prior}.
Signal recovery with a 
Laplacian regularization term
  has been used in various applications 
 \cite{Zheng_Cai2011,Elmoataz_2008,dabush2021state}.

\section{Proof of Theorem \ref{claim2}}
\label{appendix_derivation_graph_filters_noE}
For the sake of simplicity, in this and in the following appendix we use the notation $||\avec||_{\Bmat}=\avec^T\Bmat\avec$ for any vector $\avec$ and a positive semi-definite matrix $\Bmat$. In addition, we use the short notations $f_i(\Lmat)$, $f_i(\Lambdamat)$, $i=1,2$, without writing explicitly the dependency on $\hat{{\xvec}}^{(t)}$.

By substituting \eqref{opt_estimator}
 in \eqref{Q_def_approx}, we obtain
 \begin{eqnarray}
\label{LMMSE_GSP_o2}
 Q_{lin}\Big(\hat{\xvec}^{(t+1)},\hat{\xvec}^{(t)}\Big)=
 Q_{lin}\Big(\hat{\xvec}^{(t)}
	 +f_1(\Lmat)(\hat{{\xvec}}^{(t)}- {\muvec}_\xvec)
	 \nonumber\\
	+f_2(\Lmat) (\yvec-\gvec(\Lmat,\hat{\xvec}^{(t)})),\hat{{\xvec}}^{(t)}\Big)\hspace{2.4cm}\nonumber\\=
\frac{1}{2}\left|\left|(\Imat+
	 f_1(\Lmat))(\hat{{\xvec}}^{(t)}- {\muvec}_\xvec)\right.\right.\hspace{2.7cm}
	\nonumber\\\left.\left.
+f_2(\Lmat) (\yvec-\gvec(\Lmat,\hat{\xvec}^{(t)}))\right|\right|_{{\Cmat}_{\xvec\xvec}^{-1}}\hspace{2.5cm}\nonumber\\ + \frac{1}{2}\left|\left|-\Gmat(\Lmat,\hat{{\xvec}}^{(t)}) f_1(\Lmat)(\hat{{\xvec}}^{(t)}- {\muvec}_\xvec)\right.\right.\hspace{1.6cm}
	\nonumber\\\left.\left.
	+(\Imat-\Gmat(\Lmat,\hat{{\xvec}}^{(t)})f_2(\Lmat)) (\yvec-\gvec(\Lmat,\hat{\xvec}^{(t)})))\right|\right|_{{\Cmat}_{\wvec\wvec}^{-1}}
\nonumber\\
=\frac{1}{2}\left|\left|(\Imat+
	 f_1(\Lambdamat))(\hat{\tilde{\xvec}}^{(t)}- {\tilde\muvec}_\xvec)\right.\right.\hspace{2.75cm}
	\nonumber\\\left.\left.
	+f_2(\Lambdamat) (\tilde\yvec-\tilde\gvec(\Lmat,\Vmat\hat{\tilde\xvec}^{(t)}))\right|\right|_{{\Cmat}_{\tilde\xvec\tilde\xvec}^{-1}}\hspace{2.25cm}\nonumber\\ + \frac{1}{2}\left|\left| -\tilde\Gmat(\Lmat,\Vmat\hat{\tilde{\xvec}}^{(t)}) f_1(\Lambdamat)(\hat{\tilde{\xvec}}^{(t)}- \tilde{\muvec}_\xvec)\right.\right.\hspace{1.3cm}
	\nonumber\\\left.\left.
	+(\Imat-\tilde\Gmat(\Lmat,\Vmat\hat{\tilde{\xvec}}^{(t)})f_2(\Lambdamat)) (\tilde\yvec-\tilde\gvec(\Lmat,\Vmat\hat{\tilde\xvec}^{(t)}))\right|\right|_{{\Cmat}_{\tilde\wvec\tilde\wvec}^{-1}},
\end{eqnarray}
where the last equality is in the graph-frequency domain, which holds
since $\Vmat\Vmat^T = \Imat$.
By equating the derivative of \eqref{LMMSE_GSP_o2} w.r.t. $f_1(\Lambdamat)$ and $f_2(\Lambdamat)$ to zero, and using the trace operator properties, as well as the derivatives w.r.t. a diagonal matrix (see Eq. (142) in \cite{petersen2008matrix}), we obtain
\beqna
\label{der_wrt_f1}
\text{diag}\Big(\Smat_{\hat{\tilde{\xvec}}\hat{\tilde{\xvec}}}^{(t)}{\Cmat}_{\tilde\xvec\tilde\xvec}^{-1}
-\tilde\Gmat^T(\Lmat,\Vmat\hat{\tilde{\xvec}}^{(t)}) 
{\Cmat}_{\tilde\wvec\tilde\wvec}^{-1}\Smat_{\tilde{\wvec}\hat{\tilde{\xvec}}}^{(t)}\Big)
\nonumber\\
+\Big(\Smat_{\hat{\tilde{\xvec}}\hat{\tilde{\xvec}}}^{(t)}\circ
\big({\Cmat}_{\tilde\xvec\tilde\xvec}^{-1}
+\Amat_{0}\big)\Big)\text{diag}\big(f_1^*(\Lambdamat)\big)
\hspace{1.05cm}\nonumber\\
+
\Big(\Smat_{\tilde{\wvec}\hat{\tilde{\xvec}}}^{(t)}\circ\big({\Cmat}_{\tilde\xvec\tilde\xvec}^{-1}
+\Amat_{0}\big)\Big)\text{diag}\big(f_2^*(\Lambdamat)\big)=\zerovec,\hspace{0.2cm}
\eeqna
and
\beqna
\label{der_wrt_f2}
\text{diag}\Big(\Smat_{\tilde{\wvec}\hat{\tilde{\xvec}}}^{(t)}{\Cmat}_{\tilde\xvec\tilde\xvec}^{-1}
-\Smat_{\tilde{\wvec}\tilde{\wvec}}^{(t)}
{\Cmat}_{\tilde\wvec\tilde\wvec}^{-1}\tilde\Gmat(\Lmat,\Vmat\hat{\tilde{\xvec}}^{(t)}) 
\Big)
\nonumber\\
+\Big(\Smat_{\tilde{\wvec}\hat{\tilde{\xvec}}}^{(t)}\circ\big({\Cmat}_{\tilde\xvec\tilde\xvec}^{-1}
+\Amat_{0}\big)\Big)\text{diag}\big(f_1^*(\Lambdamat)\big)
\hspace{0.95cm}\nonumber\\
+
\Big(
\Smat_{\tilde{\wvec}\tilde{\wvec}}^{(t)}\circ\big({\Cmat}_{\tilde\xvec\tilde\xvec}^{-1}
+\Amat_{0}\big)\Big)\text{diag}\big(f_2^*(\Lambdamat)\big)=\zerovec,\hspace{0.1cm}
\eeqna
where 
$\Smat_{\hat{\tilde{\xvec}}\hat{\tilde{\xvec}}}^{(t)}$, $\Smat_{\tilde{\wvec}\tilde{\wvec}}^{(t)}$, and  $\Smat_{\tilde{\wvec}\hat{\tilde{\xvec}}}^{(t)}$ are defined in \eqref{Smat_x}, \eqref{Smat_y}, and \eqref{Smat_yx}, respectively, 
{\textcolor{black}{and 
\be
\label{Adef}
\Amat_{0}\define\tilde\Gmat^T(\Lmat,\Vmat\hat{\tilde{\xvec}}^{(t)}) 
	{\Cmat}_{\tilde\wvec\tilde\wvec}^{-1}
	\tilde\Gmat(\Lmat,\Vmat\hat{\tilde{\xvec}}^{(t)}).
	\ee}}
The results in \eqref{der_wrt_f1} and \eqref{der_wrt_f2} consist of a linear equation system of 
 $\text{diag}(f_i(\Lambdamat))$, $i=1,2$, and thus, 
 we obtain
\beqna
\label{filter1_noE_proof}
f_1^*(\Lambdamat,\hat{\xvec}^{(t)} ) = 
{\text{diag}}\bigg(\Big[\Smat_{\hat{\tilde{\xvec}}\hat{\tilde{\xvec}}}^{(t)}
\circ ({\Cmat}_{\tilde\xvec\tilde\xvec}^{-1} + \Amat_{0})+\Amat_1\Big]^{-1}\hspace{0.75cm}\nonumber\\
 \times \Big[\avec_2 - {\text{diag}}\big(\Smat_{\hat{\tilde{\xvec}}\hat{\tilde{\xvec}}}^{(t)}{\Cmat}_{\tilde\xvec\tilde\xvec}^{-1} -\tilde\Gmat(\Lmat,\Vmat\hat{\tilde{\xvec}}^{(t)})^{T}{\Cmat}_{\tilde\wvec\tilde\wvec}^{-1}\Smat_{\tilde{\wvec}\hat{\tilde{\xvec}}}^{(t)}\big)\Big]\bigg)
\eeqna
and
\beqna
\label{filter2_noE_proof}
f_2^*(\Lambdamat,\hat{\xvec}^{(t)} )=
-{\text{diag}}\bigg(\Big[\Smat_{\tilde{\wvec}\tilde{\wvec}}^{(t)} \circ \big({\Cmat}_{\tilde\xvec\tilde\xvec}^{-1}  + \Amat_{0}\big)\Big]^{-1}\hspace{1cm}
\nonumber\\
\times\Big[{\text{diag}}\Big(\Smat_{\tilde{\wvec}\hat{\tilde{\xvec}}}^{(t)}{\Cmat}_{\tilde\xvec\tilde\xvec}^{-1} - \Smat_{\tilde{\wvec}\tilde{\wvec}}^{(t)} {\Cmat}_{\tilde\wvec\tilde\wvec}^{-1}
\tilde\Gmat(\Lmat,\Vmat\hat{\tilde{\xvec}}^{(t)})\Big) + \avec_3\Big]\bigg),
\eeqna
where
\beqna
\label{AAA}
\Amat_1 = -
\Smat_{\hat{\tilde{\wvec}}\tilde{\xvec}}^{(t)}\circ\big({\Cmat}_{\tilde\xvec\tilde\xvec}^{-1} + \Amat_{0}\big)
\Smat_{\tilde{\wvec}\tilde{\wvec}}^{(t)}\circ\big({\Cmat}_{\tilde\xvec\tilde\xvec}^{-1}  + \Amat_{0}\big)\nonumber
\\
\times
\Smat_{\hat{\tilde{\wvec}}\tilde{\xvec}}^{(t)} \circ\big({\Cmat}_{\tilde\xvec\tilde\xvec}^{-1}  + \Amat_{0}\big),\hspace{2.65cm}
\\
\label{BBB}
\avec_2 =
\Smat_{\tilde{\wvec}\hat{\tilde{\xvec}}}^{(t)}
\Big(\Smat_{\tilde{\wvec}\tilde{\wvec}}^{(t)}\nonumber\circ\big({\Cmat}_{\tilde\xvec\tilde\xvec}^{-1} + \Amat_{0}\big)\Big)^{-1} \hspace{1.6cm}\\ \times
\text{diag}\Big(\Smat_{\tilde{\wvec}\hat{\tilde{\xvec}}}^{(t)}{\Cmat}_{\tilde\xvec\tilde\xvec}^{-1} - \Smat_{\tilde{\wvec}\tilde{\wvec}}^{(t)} {\Cmat}_{\tilde\wvec\tilde\wvec}^{-1}
\tilde\Gmat(\Lmat,\Vmat\hat{\tilde{\xvec}}^{(t)})
\Big)\hspace{-0.5cm}
\\
\label{DDD}
\avec_3 =\Smat_{\tilde{\wvec}\hat{\tilde{\xvec}}}^{(t)}\nonumber \circ\big({\Cmat}_{\tilde\xvec\tilde\xvec}^{-1} + \Amat_{0}\big)\hspace{3.1cm}\\\nonumber \times
\bigg[\Smat_{\hat{\tilde{\xvec}}\hat{\tilde{\xvec}}}^{(t)}\circ\big({\Cmat}_{\tilde\xvec\tilde\xvec}^{-1}  + \Amat_{0}\big)+\Amat_1 \bigg]^{-1}
\hspace{1.5cm}\\\nonumber\times
\bigg[\Smat_{\tilde{\wvec}\hat{\tilde{\xvec}}}^{(t)}
\circ\big({\Cmat}_{\tilde\xvec\tilde\xvec}^{-1} + \Amat_{0}\big)
\Big(\Smat_{\tilde{\wvec}\tilde{\wvec}}^{(t)}
\circ\big({\Cmat}_{\tilde\xvec\tilde\xvec}^{-1} + \Amat_{0}\big)\Big)^{-1}\hspace{-0.65cm}\\\nonumber\times
\text{diag}\Big(\Smat_{\tilde{\wvec}\hat{\tilde{\xvec}}}^{(t)}{\Cmat}_{\tilde\xvec\tilde\xvec}^{-1}-
\Smat_{\tilde{\wvec}\tilde{\wvec}}^{(t)}
{\Cmat}_{\tilde\wvec\tilde\wvec}^{-1}
\tilde\Gmat(\Lmat,\Vmat\hat{\tilde{\xvec}}^{(t)})\Big)\hspace{-0.3cm}\\-\text{diag}\Big(\Smat_{\hat{\tilde{\xvec}}\hat{\tilde{\xvec}}}^{(t)}{\Cmat}_{\tilde\xvec\tilde\xvec}^{-1} - \tilde\Gmat(\Lmat,\Vmat\hat{\tilde{\xvec}}^{(t)})^{T}{\Cmat}_{\tilde\wvec\tilde\wvec}^{-1} 
\Smat_{\tilde{\wvec}\hat{\tilde{\xvec}}}^{(t)}\Big)
\bigg].\hspace{-0.65cm}
\eeqna
By substituting 
the approximation $\Smat_{\hat{\tilde{\xvec}}\tilde{\wvec}}^{(t)}\approx\zerovec$
in \eqref{AAA}, \eqref{BBB}, and \eqref{DDD}, we obtain $\Amat_1=\zerovec
\in\mathbb{R}^{N \times N}$ and $\avec_2=\avec_3=\zerovec\in\mathbb{R}^{N}$.
By substituting these values and $\Smat_{\hat{\tilde{\xvec}}\tilde{\wvec}}^{(t)}\approx\zerovec$ in \eqref{filter1_noE_proof} and \eqref{filter2_noE_proof}, and using the notation $f_i^{\text{sGSP-MAP}}(\Lambdamat,\hat{\xvec}^{(t)} )$ instead of $f_i^*(\Lambdamat )$, $i=1,2$, we obtain the results in
\eqref{filter1_noE} and \eqref{filter2_noE}.

\section{Proof of Theorem \ref{claim1}}
\label{appendix_derivation_graph_filters}
In this appendix 
our goal is to choose the filters $f_1(\cdot)$ and $f_2(\cdot)$ such that the objective function
 $ Q_{lin}^{approx}$ is minimized on  average. 
The main course of the proof is composed of two steps; in the first step, it is shown that under the theorem conditions, 
we can replace the minimization of \eqref{obj1} by the minimization of
 \begin{eqnarray}
\label{LMMSE_GSP_o3_try2}
{\rm{E}}[  Q_{lin}^{approx}(\tilde\xvec,\hat{{\tilde\xvec}}^{(t)})] \hspace{4.75cm}
\nonumber\\=
\frac{1}{2}	{\text{trace}}\left({\Cmat}_{\tilde\xvec\tilde\xvec}(\Imat+
	 f_1(\Lambdamat)){\Cmat}_{\tilde\xvec\tilde\xvec}^{-1}((\Imat
	+ f_1(\Lambdamat))\right)    \nonumber\\ 
+\frac{1}{2}	{\text{trace}}\left(
{\Cmat}_{\tilde\wvec\tilde\wvec}
f_2(\Lambdamat) {\Cmat}_{\tilde\xvec\tilde\xvec}^{-1}f_2(\Lambdamat)\right)
\hspace{1.5cm}
\nonumber\\
+\frac{1}{2}	{\text{trace}}\left({\Cmat}_{\tilde\xvec\tilde\xvec}f_1(\Lambdamat)
 \Amat_{0} f_1(\Lambdamat)\right)\hspace{2.1cm}
	\nonumber\\
	+	{\text{trace}}\left(\frac{1}{2}\Imat-\tilde\Gmat(\Lmat,\Vmat\hat{\tilde{\xvec}}^{(t)})f_2(\Lambdamat)\right) \hspace{1.25cm}
\nonumber\\
	+\frac{1}{2}	{\text{trace}}\left({\Cmat}_{\tilde\wvec\tilde\wvec}f_2(\Lambdamat)\Amat_{0} f_2(\Lambdamat)\right), \hspace{1.85cm}
\eeqna
where ${\mathcal{D}}_N$ is the set of diagonal matrices of size $N \times N$ and
$\Amat_{0}$ is defined in \eqref{Adef}.
	In the second step, it is shown that this minimization w.r.t. the graph filters results in \eqref{filter1} and \eqref{filter2}.

First, we note that based on the assumption that $\hat{\xvec}^{(t)}$ is close enough to $\xvec$,  we replace $\tilde\gvec(\Lmat,\Vmat\hat{\tilde\xvec}^{(t)}))$ by $\tilde\gvec(\Lmat,\Vmat\tilde\xvec)$, and 
$\hat{\tilde{\xvec}}^{(t)}- \tilde{\muvec}_\xvec$ by
$\hat{\tilde{\xvec}}- \tilde{\muvec}_\xvec$ in \eqref{LMMSE_GSP_o2}, but keep $\tilde\Gmat(\Lmat,\Vmat\hat{\tilde{\xvec}}^{(t)})$ evaluated at $\hat{\tilde{\xvec}}^{(t)}$ (similar to  the rationale behind the first-order approximation in the Gauss-Newton method \cite{Bell_Cathey1993}), to obtain the approximation
 \begin{eqnarray}
\label{LMMSE_GSP_o2_approx}
 Q_{lin}^{approx}(\tilde{\xvec},\hat{\tilde\xvec}^{(t)})=\hspace{5.25cm}\nonumber\\
\frac{1}{2}\left|\left|(\Imat+
	 f_1(\Lambdamat)({\tilde{\xvec}}- {\tilde\muvec}_\xvec)
	+f_2(\Lambdamat) (\tilde\yvec-\tilde\gvec(\Lmat,\Vmat\hat{\tilde\xvec}))\right|\right|_{{\Cmat}_{\tilde\xvec\tilde\xvec}^{-1}}\nonumber\\ + \frac{1}{2}\left|\left| -\tilde\Gmat(\Lmat,\Vmat\hat{\tilde{\xvec}}^{(t)}) f_1(\Lambdamat)({\tilde{\xvec}}- \tilde{\muvec}_\xvec)\right.\right.\hspace{3cm}
	\nonumber\\\left.\left.
	+(\Imat-\tilde\Gmat(\Lmat,\Vmat\hat{\tilde{\xvec}}^{(t)})f_2(\Lambdamat)) (\tilde\yvec-\tilde\gvec(\Lmat,\Vmat{\tilde\xvec}))\right|\right|_{{\Cmat}_{\tilde\wvec\tilde\wvec}^{-1}}.
\eeqna
Thus,  we minimize  the expected objective function from \eqref{LMMSE_GSP_o2_approx}, under the assumption that $\tilde{\xvec}$ and $\tilde{\wvec}$ are independent, and treating $\tilde\Gmat(\Lmat,\Vmat\hat{\tilde{\xvec}}^{(t)})$ as a deterministic matrix:
 \begin{eqnarray}
\label{LMMSE_GSP_o3_try}
\{f_1^*(\Lambdamat ),f_2^*(\Lambdamat )\}=
 \arg\hspace{-0.2cm}\min_{f_1(\Lambdamat ) ,f_2(\Lambdamat ) \in {\mathcal{D}}_N}\hspace{-0.2cm} {\rm{E}}[  Q_{lin}^{approx}(\tilde{\xvec},\hat{\tilde\xvec}^{(t)})],
 \eeqna
 where ${\rm{E}}[  Q_{lin}^{approx}(\tilde{\xvec},\hat{\tilde\xvec}^{(t)})]$ is defined in \eqref{LMMSE_GSP_o3_try2}.
Since the minimization is separable w.r.t. $f_1$ and $f_2$, we can solve it independently. 
Thus,
 \begin{eqnarray}
 \label{f1_term}
f_1^*(\Lambdamat )=
 \arg\min_{f_1(\Lambdamat ) \in {\mathcal{D}}_N}
 	\sum_{n=1}^N [f_1(\Lambdamat)]_{n,n}
\hspace{1.5cm}\nonumber\\
+
\frac{1}{2}	\sum_{n=1}^N \sum_{k=1}^N 
[{\Cmat}_{\tilde\xvec\tilde\xvec}]_{n,k} [{\Cmat}_{\tilde\xvec\tilde\xvec}^{-1}]_{k,n} [f_1(\Lambdamat)]_{k,k}
[f_1(\Lambdamat))]_{n,n}
\nonumber\\
+
\frac{1}{2}	\sum_{n=1}^N\sum_{k=1}^N
[{\Cmat}_{\tilde\xvec\tilde\xvec}]_{n,k}[\Amat_{0}]_{k,n}
[f_1(\Lambdamat)]_{k,k}
 [ f_1(\Lambdamat)]_{n,n},
	\hspace{0.4cm}\end{eqnarray}
where $\Amat_{0}$ is defined in \eqref{Adef}.
By equating the derivative of \eqref{f1_term} w.r.t. $[f_1(\Lambdamat)]_{l,l}$ to zero, one obtains
 \begin{eqnarray}
 \label{derivative1}
0=	1+\sum_{n=1}^N
[{\Cmat}_{\tilde\xvec\tilde\xvec}]_{n,l}
[{\Cmat}_{\tilde\xvec\tilde\xvec}^{-1}]_{l,n}
	 [f_1^*(\Lambdamat)]_{n,n}
\hspace{0.7cm}\nonumber\\
+	\sum_{n=1}^N
[{\Cmat}_{\tilde\xvec\tilde\xvec}]_{n,l}[\Amat_{0}]_{l,n}
 [ f_1^*(\Lambdamat)]_{n,n}, ~\forall l=1,\ldots,N,
	\eeqna
	which results in
	 \begin{eqnarray}
	 \label{f1opt}
{\text{diag}}(f_1^*(\Lambdamat))=-
({\Cmat}_{\tilde\xvec\tilde\xvec}
\circ {\Cmat}_{\tilde\xvec\tilde\xvec}^{-1}
+{\Cmat}_{\tilde\xvec\tilde\xvec}
\circ \Amat_{0})^{-1}
\onevec.
\end{eqnarray}
By applying the diag operator on both sides of \eqref{f1opt}, substituting \eqref{Adef}, 
and using the notation $f_1^{\text{GSP-MAP}}(\Lambdamat,\hat{\xvec}^{(t)} )$ instead of $f_1^*(\Lambdamat )$,  we obtain
the graph filter in \eqref{filter1}.

Similarly, the optimization w.r.t. the filter $f_2(\Lambdamat )$ is 
 \begin{eqnarray}
 \label{optf2}
f_2^*(\Lambdamat )=
 \arg\min_{f_2(\Lambdamat ) \in {\mathcal{D}}_N}\hspace{2.9cm}\nonumber\\
 \frac{1}{2}	\sum_{n=1}^N\sum_{k=1}^N
[{\Cmat}_{\tilde\wvec\tilde\wvec}]_{n,k}[{\Cmat}_{\tilde\xvec\tilde\xvec}^{-1}]_{k,n}
[f_2(\Lambdamat)]_{k,k}
 [ f_2(\Lambdamat)]_{n,n}
 \nonumber\\
 -\sum_{n=1}^N [\tilde\Gmat(\Lmat,\Vmat\hat{\tilde{\xvec}}^{(t)})]_{n,n}[f_2(\Lambdamat)]_{n,n}
\hspace{2.3cm}\nonumber\\
+
\frac{1}{2}	\sum_{n=1}^N\sum_{k=1}^N
[{\Cmat}_{\tilde\wvec\tilde\wvec}]_{n,k}[\Amat_{0}]_{k,n}
[f_2(\Lambdamat)]_{k,k}
 [ f_2(\Lambdamat)]_{n,n}.
	\end{eqnarray}
By equating the derivative of \eqref{optf2} w.r.t. $[f_2(\Lambdamat)]_{l,l}$ to zero, one obtains
 \begin{eqnarray}
 \label{derivative2}
0=
	\sum_{n=1}^N
[{\Cmat}_{\tilde\wvec\tilde\wvec}]_{n,l}[{\Cmat}_{\tilde\xvec\tilde\xvec}^{-1}]_{l,n}
 [ f_2^*(\Lambdamat)]_m\hspace{2.5cm}
 \nonumber\\
-	[\tilde\Gmat(\Lmat,\Vmat\hat{\tilde{\xvec}}^{(t)})]_{l,l}
+	\sum_{n=1}^N
[{\Cmat}_{\tilde\wvec\tilde\wvec}]_{n,l}[\Amat_{0}]_{l,n}
 [ f_2^*(\Lambdamat)]_m, 
	\eeqna
	$\forall l=1,\ldots,N$,
	which results in
	 \begin{eqnarray}
	 \label{f2opt}
{\text{diag}}(f_2^*(\Lambdamat))\hspace{5.75cm}
\nonumber\\=({\Cmat}_{\tilde\wvec\tilde\wvec} \circ {\Cmat}_{\tilde\xvec\tilde\xvec}^{-1} + {\Cmat}_{\tilde\wvec\tilde\wvec} \circ\Amat_{0})^{-1}{\text{diag}}(\tilde\Gmat(\Lmat,\Vmat\hat{\tilde{\xvec}}^{(t)})).
	\end{eqnarray}
By applying the diag operator on both sides of \eqref{f2opt},  substituting \eqref{Adef}, and using the notation $f_2^{\text{GSP-MAP}}(\Lambdamat,\hat{\xvec}^{(t)} )$ instead of $f_2^*(\Lambdamat )$,  we obtain
	the graph filter in \eqref{filter2}. 

\end{document}

%% file: myshorts.tex
\newcommand{\avec}{{\bf{a}}}

\newcommand{\bvec}{{\bf{b}}}

\newcommand{\pvec}{{\bf{p}}}

\newcommand{\yvec}{{\bf{y}}}

\newcommand{\wvec}{{\bf{w}}}
\newcommand{\xvec}{{\bf{x}}}

\newcommand{\vvec}{{\bf{v}}}
\newcommand{\gvec}{{\bf{g}}}

\newcommand{\onevec}{{\bf{1}}}
\newcommand{\zerovec}{{\bf{0}}}

\newcommand{\Lambdamat}{{\bf{\Lambda}}}

\newcommand{\Amat}{{\bf{A}}}
\newcommand{\Bmat}{{\bf{B}}}
\newcommand{\Cmat}{{\bf{C}}}
\newcommand{\Dmat}{{\bf{D}}}

\newcommand{\Gmat}{{\bf{G}}}

\newcommand{\Imat}{{\bf{I}}}
\newcommand{\Lmat}{{\bf{L}}}

\newcommand{\Smat}{{\bf{S}}}

\newcommand{\Vmat}{{\bf{V}}}
\newcommand{\Wmat}{{\bf{W}}}

\newcommand{\define}{\stackrel{\triangle}{=}}

\def\muvec{{\mbox{\boldmath $\mu$}}}

\newcommand{\be}{\begin{equation}}
\newcommand{\ee}{\end{equation}}
\newcommand{\beqna}{\begin{eqnarray}}
\newcommand{\eeqna}{\end{eqnarray}}
